\documentclass[pra,aps,nopacs,onecolumn,twoside,superscriptaddress]{revtex4}

\usepackage{multirow, makecell}
\usepackage{amsmath,amsfonts,amssymb,caption,color,epsfig,graphics,graphicx,hyperref,latexsym,mathrsfs,revsymb,theorem,url,verbatim,epstopdf,enumerate}
\usepackage{amsmath,amsfonts,amssymb,caption,hyperref,color,epsfig,graphics,graphicx,latexsym,mathrsfs,revsymb,theorem,url,verbatim,epstopdf,cleveref}
\usepackage{fontenc}
\usepackage{cases}
\usepackage[T1]{fontenc}
\hypersetup{colorlinks,linkcolor={blue},citecolor={blue},urlcolor={red}}
\usepackage{lmodern} 
\usepackage{tipa}
\usepackage{textcomp} % 在导言区引入包

\newtheorem{definition}{Definition}
\newtheorem{proposition}[definition]{Proposition}
\newtheorem{lemma}[definition]{Lemma}

\newtheorem{theorem}[definition]{Theorem}
\newtheorem{corollary}[definition]{Corollary}
\newtheorem{conjecture}[definition]{Conjecture}

\newtheorem{remark}[definition]{Remark}
\newtheorem{example}[definition]{Example}
\newtheorem{question}[definition]{Question}
\newtheorem{memo}[definition]{Memo}

\def\squareforqed{\hbox{\rlap{$\sqcap$}$\sqcup$}}
\def\qed{\ifmmode\squareforqed\else{\unskip\nobreak\hfil
		\penalty50\hskip1em\null\nobreak\hfil\squareforqed
		\parfillskip=0pt\finalhyphendemerits=0\endgraf}\fi}
\def\endenv{\ifmmode\;\else{\unskip\nobreak\hfil
		\penalty50\hskip1em\null\nobreak\hfil\;
		\parfillskip=0pt\finalhyphendemerits=0\endgraf}\fi}
\newenvironment{proof}{\noindent \textbf{{Proof.~} }}{\qed}
\def\Dbar{\leavevmode\lower.6ex\hbox to 0pt
	{\hskip-.23ex\accent"16\hss}D}
\makeatletter
\def\url@leostyle{%
	\@ifundefined{selectfont}{\def\UrlFont{\sf}}{\def\UrlFont{\small\ttfamily}}}
\makeatother
\def\bcj{\begin{conjecture}}
	\def\ecj{\end{conjecture}}
\def\bcr{\begin{corollary}}
	\def\ecr{\end{corollary}}
\def\bd{\begin{definition}}
	\def\ed{\end{definition}}
\def\bea{\begin{eqnarray}}
	\def\eea{\end{eqnarray}}
\def\bem{\begin{enumerate}}
	\def\eem{\end{enumerate}}
\def\bex{\begin{example}}
	\def\eex{\end{example}}
\def\bim{\begin{itemize}}
	\def\eim{\end{itemize}}
\def\bl{\begin{lemma}}
	\def\el{\end{lemma}}
\def\bma{\begin{bmatrix}}
	\def\ema{\end{bmatrix}}
\def\bpf{\begin{proof}}
	\def\epf{\end{proof}}
\def\bpp{\begin{proposition}}
	\def\epp{\end{proposition}}
\def\bqu{\begin{question}}
	\def\equ{\end{question}}
\def\br{\begin{remark}}
	\def\er{\end{remark}}
\def\bt{\begin{theorem}}
	\def\et{\end{theorem}}
\def\bmm{\begin{memo}}
	\def\emm{\end{memo}}

\def\btb{\begin{tabular}}
	\def\etb{\end{tabular}}

	\newcommand{\nc}{\newcommand}
	
	\def\a{\alpha}
	\def\b{\beta}
	\def\g{\gamma}
	
	\def\e{\epsilon}
	
	\def\z{\zeta}

	\def\l{\lambda}

	\def\r{\rho}
	\def\s{\sigma}

	\def\ps{\psi}

	\def\G{\Gamma}

	\nc{\bbA}{\mathbb{A}} \nc{\bbB}{\mathbb{B}} \nc{\bbC}{\mathbb{C}}
	\nc{\bbD}{\mathbb{D}} \nc{\bbE}{\mathbb{E}} \nc{\bbF}{\mathbb{F}}
	\nc{\bbG}{\mathbb{G}} \nc{\bbH}{\mathbb{H}} \nc{\bbI}{\mathbb{I}}
	\nc{\bbJ}{\mathbb{J}} \nc{\bbK}{\mathbb{K}} \nc{\bbL}{\mathbb{L}}
	\nc{\bbM}{\mathbb{M}} \nc{\bbN}{\mathbb{N}} \nc{\bbO}{\mathbb{O}}
	\nc{\bbP}{\mathbb{P}} \nc{\bbQ}{\mathbb{Q}} \nc{\bbR}{\mathbb{R}}
	\nc{\bbS}{\mathbb{S}} \nc{\bbT}{\mathbb{T}} \nc{\bbU}{\mathbb{U}}
	\nc{\bbV}{\mathbb{V}} \nc{\bbW}{\mathbb{W}} \nc{\bbX}{\mathbb{X}}
	\nc{\bbZ}{\mathbb{Z}}
	
	\nc{\bA}{{\bf A}} \nc{\bB}{{\bf B}} \nc{\bC}{{\bf C}}
	\nc{\bD}{{\bf D}} \nc{\bE}{{\bf E}} \nc{\bF}{{\bf F}}
	\nc{\bG}{{\bf G}} \nc{\bH}{{\bf H}} \nc{\bI}{{\bf I}}
	\nc{\bJ}{{\bf J}} \nc{\bK}{{\bf K}} \nc{\bL}{{\bf L}}
	\nc{\bM}{{\bf M}} \nc{\bN}{{\bf N}} \nc{\bO}{{\bf O}}
	\nc{\bP}{{\bf P}} \nc{\bQ}{{\bf Q}} \nc{\bR}{{\bf R}}
	\nc{\bS}{{\bf S}} \nc{\bT}{{\bf T}} \nc{\bU}{{\bf U}}
	\nc{\bV}{{\bf V}} \nc{\bW}{{\bf W}} \nc{\bX}{{\bf X}}
	\nc{\bZ}{{\bf Z}}
	
	\nc{\as}{{\cal AS}}
	\nc{\app}{{\cal AP}}
	\nc{\ar}{{\cal AR}}
	\nc{\bp}{{\cal BP}}
	\nc{\dbp}{{\cal DBP}}
\nc{\ew}{{\cal EW}}
\nc{\dew}{{\cal DEW}}
\nc{\ndew}{{\cal NDEW}}
\nc{\conv}{{\text{Conv}}}

	\nc{\cA}{{\cal A}} \nc{\cB}{{\cal B}} \nc{\cC}{{\cal C}}
	\nc{\cD}{{\cal D}} \nc{\cE}{{\cal E}} \nc{\cF}{{\cal F}}
	\nc{\cG}{{\cal G}} \nc{\cH}{{\cal H}} \nc{\cI}{{\cal I}}
	\nc{\cJ}{{\cal J}} \nc{\cK}{{\cal K}} \nc{\cL}{{\cal L}}
	\nc{\cM}{{\cal M}} \nc{\cN}{{\cal N}} \nc{\cO}{{\cal O}}
	\nc{\cP}{{\cal P}} \nc{\cQ}{{\cal Q}} \nc{\cR}{{\cal R}}
	\nc{\cS}{{\cal S}} \nc{\cT}{{\cal T}} \nc{\cU}{{\cal U}}
	\nc{\cV}{{\cal V}} \nc{\cW}{{\cal W}} \nc{\cX}{{\cal X}}
	\nc{\cZ}{{\cal Z}}
	
	\nc{\cpp}{{\cal PP}}
	
	\nc{\hA}{{\hat{A}}} \nc{\hB}{{\hat{B}}} \nc{\hC}{{\hat{C}}}
	\nc{\hD}{{\hat{D}}} \nc{\hE}{{\hat{E}}} \nc{\hF}{{\hat{F}}}
	\nc{\hG}{{\hat{G}}} \nc{\hH}{{\hat{H}}} \nc{\hI}{{\hat{I}}}
	\nc{\hJ}{{\hat{J}}} \nc{\hK}{{\hat{K}}} \nc{\hL}{{\hat{L}}}
	\nc{\hM}{{\hat{M}}} \nc{\hN}{{\hat{N}}} \nc{\hO}{{\hat{O}}}
	\nc{\hP}{{\hat{P}}} \nc{\hR}{{\hat{R}}} \nc{\hS}{{\hat{S}}}
	\nc{\hT}{{\hat{T}}} \nc{\hU}{{\hat{U}}} \nc{\hV}{{\hat{V}}}
	\nc{\hW}{{\hat{W}}} \nc{\hX}{{\hat{X}}} \nc{\hZ}{{\hat{Z}}}
	
	\nc{\hn}{{\hat{n}}}
	
	\def\diag{\mathop{\rm diag}}
	\def\dim{\mathop{\rm Dim}}

	\def\max{\mathop{\rm max}}
	\def\min{\mathop{\rm min}}

	\def\tr{\mathop{\rm Tr}}

	\def\span{\mathop{\rm Span}}

	\def\dg{\dagger}

	\def\ra{\rightarrow}

	\newcommand{\bra}[1]{\langle#1|}
	\newcommand{\ket}[1]{|#1\rangle}
	\newcommand{\proj}[1]{| #1\rangle\!\langle #1 |}

	\newcommand{\norm}[1]{\lVert#1\rVert}

	\def\Dbar{\leavevmode\lower.6ex\hbox to 0pt
		{\hskip-.23ex\accent"16\hss}D}
\begin{document}
	%\Large

\title{Spectral characterizations of entanglement witnesses}

\date{\today}

\author{Zhiwei Song}\email[]{zhiweisong@buaa.edu.cn}
\affiliation{LMIB(Beihang University), Ministry of education, and School of Mathematical Sciences, Beihang University, Beijing 100191, China}

\author{Lin Chen}\email[]{linchen@buaa.edu.cn (corresponding author)}
\affiliation{LMIB(Beihang University), Ministry of education, and School of Mathematical Sciences, Beihang University, Beijing 100191, China}

%\Large

\begin{abstract}
We present a systematic investigation of the spectral properties of entanglement witnesses (EWs). Specifically, we analyze the infimum and supremum of the largest eigenvalue, the smallest eigenvalue, the negativity (defined as the absolute value of the sum of negative eigenvalues), and the squared Frobenius norm of a unit-trace (normalized) entanglement witness, along with the conditions under which these values are attained. Our study provides distinct characterizations for decomposable (DEWs) and nondecomposable entanglement witnesses (NDEWs). While these two classes share many spectral similarities, we reveal a fundamental divergence by proving that the infimum of the smallest eigenvalue
can be attained by DEWs, yet remains strictly unattainable for all NDEWs. 
We apply the results to provide necessary conditions for an EW to possess a mirrored EW.
Furthermore, we demonstrate the superior detection capability of NDEWs by proving that any non-positive-transpose (NPT) state beyond the two-qubit and qubit-qutrit systems can be detected by an NDEW.
\end{abstract}

\maketitle


%\tableofcontents

\section{Introduction}
\label{sec:int}
Characterization of entangled states is an important issue in
quantum information theory, from both a theoretical and an experimental perspective. In particular, it is of primary importance
to test whether a given quantum state is entangled \cite{horodecki2009quantum}.
For low-dimensional systems, the celebrated Peres–Horodecki criterion \cite{peres1996,horodecki2001separability} implies that a state within two-qubit or qubit-qutrit system is separable if and only if its partial transpose is positive, i.e., the state is PPT. However, for higher-dimensional systems, the condition is only necessary, meaning the existence of PPT entangled states. It turns out that determining the separability is an NP-hard problem \cite{gurvits2003classical}. Despite this, there are considerable efforts to distinguish entangled states from separable states, see for a review of entanglement detection \cite{guhne2009entanglement}. 

One of such powerful approaches is the so-called entanglement witnesses (EWs). 
 EWs are Hermitian operators that produce non-negative expectation values on all separable states. So if the expectation value on some state
is negative, then that state is guaranteed to be entangled.
This provides a tool to verify entanglement in experiments
since they are in principle measurable observables. In the past decades, EWs have been studied from various points of view, including connections with Bell inequalities \cite{hyllus2005relations}, optimization of witnesses \cite{lewenstein2000optimization1}, the nonlinear extensions \cite{guhne2006nonlinear} and so on. We refer readers to \cite{Chru2014Entanglement} for a recent review on EWs. Nevertheless, not every EW is useful for detecting a given entangled state \cite{horodecki1997}. The decomposable EWs cannot detect PPT entangled state (PPTES) and, therefore, 
nondecomposable EWs  have to be constructed to detect
bound entangled states. The PPT entangled states have been studied in terms of Schmidt number \cite{sanpera2001schmidt}, unextendible product basis \cite{sollid2011unextendible}, the total characterization of two-qutrit rank-four PPT entangled states \cite{chen2011description,Chen2012Equivalence1}, entanglement distillation and \cite{horodecki1997inseparable,chen2011distillability2}.
Hence, understanding the mathematical structure of EWs has become a more important step for the above theoretical problems and applications.

The spectral properties of EWs have been extensively studied in the literature. Several classes of decomposable EWs have been constructed based on a specific spectral condition \cite{chruscinski2009spectral1,chruscinski2009spectral}.
It is known that an EW $W\in \cM_m(\bbC)\otimes \cM_n(\bbC)$ has at most $(m-1)(n-1)$ negative eigenvalues \cite{sarbicki2008spectral}, with this bound being attainable by decomposable EWs \cite{johnston2013non}. Further, the inertia of EWs has been partially studied for NPT states within qubit-qudit and two-qutrit systems \cite{rana2013negative,shen2020inertias, feng2024inertia,liang2024inertia}. Additionally, eigenvalue inequalities for EWs have been explored, particularly concerning the lower bound of the ratio 
$\frac{\l_{\text{min}}(W)}{\l_{\text{max}}(W)}$ \cite{johnston2010family,johnston2012norms}.
In \cite{johnston2018inverse}, by applying the established results of absolutely separable and PPT states, a necessary and sufficient spectral condition was established for qubit-qudit EWs, with necessary conditions further extended to higher-dimensional systems.

In this paper, we shall investigate more spectral properties of EWs in arbitrary-dimensional systems. Specifically, we analyze the supremum and infimum of $\tr(W^2)$, $\l_{\text{max}}(W)$, $\l_{\text{min}}(W)$ and $\cN(W)$ for a normalized EW $W$. Another aspect of our work involves determining whether these theoretical bounds are attainable. 
Our characterizations are presented separately for decomposable entanglement witnesses (DEWs) and nondecomposable entanglement witnesses (NDEWs), and our main findings are summarized in Table \ref{ta:eigenvalues}.
For DEWs, we establish necessary and sufficient conditions for the attainment of these infima and suprema. Additionally, by introducing two new absolutely PPT states, we determine the infimum of the sum of the two and three smallest eigenvalues of a normalized DEW. Our analysis reveals that while many of the aforementioned properties are consistent between DEWs and NDEWs, there are some notable differences. In particular, we find that the infimum of $\l_{\text{min}}(W)$ can be attained by DEWs, but remains unattainable for NDEWs.

\begin{table}
	\begin{center}   
		\caption{Suprema and infima on the largest eigenvalue $\l_1$, smallest eigenvalue $\l_{mn}$, the negativity $\cN(W)$, and $\tr(W^2)$ in terms of trace-one (normalized) DEWs and NDEWs in $\cM_m(\bbC)\otimes \cM_n(\bbC)$ with $m\le n$. These results are contained in Theorems \ref{th:lambdaIN[-1/2,1)} and \ref{NDEW}. The question mark "?" denotes the unsolved cases.}  
		\label{ta:eigenvalues}
		\begin{tabular}{|c|c|c|c|c|}   
			\hline    
			& DEWs, sup, &DEWs, inf &NDEWs, sup & NDEWs, inf  \\      
			\hline    
			$\l_1$ 
			&
			\makecell[l]{
				$1$, not attainable} 
			& 
			\makecell[l]{
				${1\over mn-1}$, not attainable} 
			& 
			\makecell[l]{
				1, not attainable
					} 
			& 
			\makecell[l]{
				?, not attainable\\  
			} 
			\\   
			\hline    
			$\l_{mn}$ 
			& 
			\makecell[l]{
				$0$, not attainable \\  
				} 
			& 
			\makecell[l]{
				$-\frac{1}{2}$, attainable \\  
				} 
			& 
			\makecell[l]{
				0, not attainable} 
			& 
			\makecell[l]{
				$-\frac{1}{2}$, not attainable}
			\\    
			\hline   
			$\cN(W)$ 
			& 
			\makecell[l]{
				${m-1\over2}$, attainable} 
			& 
			\makecell[l]{
				$0$, not attainable} 
			& 
			\makecell[l]{
				$\frac{m-1}{2}$ for $m=2$ and $m=n\ge 3$, \\
				not attainable for $m=n=3$, \\? for other cases
			}  
			& 
			\makecell[l]{
				0, not attainable
			}
				\\    
			\hline   
			$\tr(W^2)$ 
			& 
			\makecell[l]{
				$1$, attainable} 
			& 
			\makecell[l]{
				$\frac{1}{mn-1}$, not attainable} 
			& 
			\makecell[l]{
				$1$, ?
			}  
			& 
			\makecell[l]{
				?, not attainable
			}
			\\
			\hline   
		\end{tabular}   
	\end{center}   
\end{table}

The rest of this paper is organized as follows. In Sec. \ref{sec:pre}, we introduce the basic knowledge and facts relevant to this paper. In Sec. \ref{sec:res=eigenvalues}, we characterize the spectral properties for DEWs and NDEWs in Theorem \ref{th:lambdaIN[-1/2,1)} and \ref{NDEW}, respectively. In Section \ref{Power}, we apply our results to provide some necessary conditions for an EW to possess a mirrored EW. We also show that NDEWs have a strong detection capability in Theorem \ref{NPTNDEW}.
Finally, we conclude in Sec. \ref{sec:con}.

\section{Preliminaries}
\label{sec:pre}
In this section, we first establish the basic notations used throughout this paper. We then review fundamental concepts related to the separability problem and entanglement witnesses (EWs). Finally, we introduce the notions of absolutely separable (AS) and positive partial transpose (PPT) states, which serve as key tools for characterizing the spectral properties of EWs.

\subsection{Basic notations and results}
Let  $\ket{i}\in \bbC^n$ denote the standard basis vector whose $i$-th component is one and all other entries are zero. Given a vector $\ket{a}$, we denote $\ket{a^*}$ as its complex conjugate and $\bra{a}$ as its conjugate transpose.
For a matrix $M$, we denote $\cR(M)$ and $\cK(M)$ as the range and kernel space of matrix $M$, respectively. For a Hermitian matrix $M$, we refer to $\l(M):=(\l_1(M),\cdots,\l_{n}(M))$ as the eigenvalue vector of $M$, arranged in non-increasing order. We shall take $\l$ as $\l(M)$ and $\l_j$ as $\l_j(M)$ when $M$ is clear from the context. The family of Schatten p-norm, is defined for $p\in [1,\infty]$ by $||X||_p:=(\tr(X^\dg X)^\frac{p}{2})^{\frac{1}{p}}$, with $||X||_1,||X||_2$ being the trace norm and Frobenius norm, respectively. We say that a matrix $M\ge 0$ if $M$ is positive semidefinite. We define  $\cN(M):=\frac{||M||_1-1}{2}$ as the absolute value of the sum of its negative eigenvalues. This definition guarantees that $\cN(M)$ is convex. 
Given two vectors $x,y\in \bbR^n$, we say that $y$ majorizes $x$, denote as $y\succ x$ (or $x\prec y$)  if 
$\sum_{i=1}^k x_i^{\downarrow}\le \sum_{i=1}^k y_i^{\downarrow}$, for any $k=1,\cdots,n-1$ and $\sum_{i=1}^n x_i^{\downarrow}=\sum_{i=1}^n y_i^{\downarrow}$ ($x^{\downarrow}$ denotes the non-increasing order of $x$).

The following lemma summarizes some basic eigenvalue inequalities for Hermitian matrices. We refer readers to \cite{horn2012matrix1} for more details.
\begin{lemma}
	\label{ineqe}
	Let $A,B$ be two order-$n$ Hermitian matrices. Then
	
	(i) $\l(A+B)\prec \l(A)+\l(B)$.

	(ii) $\l_i(A)+\l_n(B)\le \l_i(A+B), i=1,\cdots,n$
	with equality if and only if there is a nonzero vector $x$ such that $Ax=\l_i(A)x$, $Bx=\l_n(B)x$ and $(A+B)x=\l_i(A+B)x$. 
	
	(iii) $\sum_{i=1}^n (\l_i(A)-\l_i(B))^2\le ||B-A||_2^2$.
	
	(iv) $\tr(A\cdot B)\ge \sum_{i=1}^n \l_i(A)\cdot\l_{n-i+1}(B)$.
	
	(v) Let $H$ be an $n\times n$ Hermitian matrix  partitioned as $H=\bma A&&B\\B^*&&C\ema$, where $A$ is a principal submatrix of order-$m$ with $m\le n$. Then $\l_{k+n-m}(H)\le \l_k(A)\le \l_k(H)$ for any $k=1,\cdots,m$.
\end{lemma}

Given a bipartite matrix $M\in \cM_m(\bbC)\otimes \cM_n(\bbC)$, we denote $M^\G:=(T\otimes I)M$ as its partial transpose, where $T$ and $I$ denote the transpose map and identity map respectively. Unlike the global transpose map, the eigenvalues of the matrix can change after applying the partial transpose.
The following result gives a full characterization of the eigenvalues of the partial transpose of a pure state.
\begin{lemma}
	\label{ept}
	\cite{rana2013negative}
	Let $\ket{\phi}\in \bbC^m\otimes \bbC^n$ ($m\le n$) be written in its Schmidt decomposition as $\ket{\phi}=\sum_{j=1}^m a_j\ket{b_j}\otimes \ket{c_j}$. Then the eigenvalues of $\proj{\phi}^\Gamma$ are $a_j^2$ for $1\le j\le m$, and $\pm a_ia_j$ for $1\le i<j\le m$, and $m(n-m)$ extra zero eigenvlaues. Moreover, the eigenspace corresponding to the eigenvalue $\pm a_ia_j$ is spanned by $\ket{b_i^*}\otimes \ket{c_j}\pm\ket{b_j^*}\otimes \ket{c_i}$, the eigenspace corresponding to the eigenvalue $a_i^2$ is spanned by $\ket{b_i^*}\otimes \ket{c_i}$. 
\end{lemma}

\subsection{Block-positive matrices and quantum entanglement}
A quantum state $\r\in \cM_m(\bbC)\otimes \cM_n(\bbC)$
is called separable if it can be written as $\r=\sum_i p_i\proj{v_i}\otimes \proj{w_i}$, with $p_i\ge 0$, $\sum_i p_i=1$, $\ket{v_i}\in\bbC^m$ and $\ket{w_i}\in\bbC^n$.
The famous positive-partial-transpose (PPT) criterion states that if $\r$ is separable, then $\r^\G\ge 0$. But this criterion is sufficient only if $mn\le 6$. Some examples of PPT entangled states (PPTES) in higher-dimensional systems were constructed based on the range criterion \cite{horodecki1997separability}. According to this criterion, if a state $\r$ is separable then there must exist a set of product vectors $\ket{e_k,f_k}$ that span $\cR(\r)$ such that the set of vectors $\ket{e_k^*,f_k}$ span $\cR(\r^\G)$. Further, a PPTES $\r$ is called to be edge state if there is
no product vector $\ket{a,b}\in \cR(\r)$ such that $\ket{a^*,b}\in \cR(\r^\G)$. By definition, edge states lie
on the boundary between NPT and PPT states and violate the range criterion in an extreme manner. It is also known that any edge state has deficient rank, otherwise it lies in the interior of the set of PPT states \cite{grabowski2005geometry}.

A bipartite Hermtian matirx $M\in \cM_m(\bbC)\otimes \cM_n(\bbC)$ is called an $m\times n$ block-positive matrix if  $\bra{a,b}W\ket{a,b}\ge 0$ holds for any product vector $\ket{a,b}\in \bbC^m\otimes\bbC^n$. The set of $m\times n$ block-positive matrices is denoted as $\bp_{m,n}$. By definition, the partial transpose of a block-positive matrix $W$ is also block-positive, since $\bra{a,b}W^\G\ket{a,b}=\bra{a^*,b}W\ket{a^*,b}\ge 0$.
The following is another property of block-positive matrices.
\begin{lemma}
	\label{proj}
Let $W=[W_{i,j}]_{i,j=1}^m\in \bp_{m,n}$. 

(i) If $W_{k,k}=0$, then $W_{k,j}=W_{j,k}=0$ for any $j$.

(ii)  If for each $i$, the $k$-th diagonal entry of $W_{i,i}$ vanishes, then the entire $k$-th row and column of every block $W_{i,j}$ must be zero.
\end{lemma}
\begin{proof}
(i)	Without loss of generality, suppose $W_{1,1}=0$ and $W_{1,2}\neq 0$. Then there exist $a_1,a_2\in \bbC$ such that the matrix $|a_2|^2 W_{2,2}+a_1^*a_2W_{1,2}+a_1a_2^*W_{1,2}^\dg$ is not positive semidefinte. Consequently, let $\ket{a}:=(a_1,a_2,0,\cdots,0)\in \bbC^m$. Then 
the order-$n$ matrix $(\bra{a}\otimes I_n) W(\ket{a}\otimes I_n)$ is not positive semidefinite, which contradicts the block-positivity of $W$.  Hence $W_{1,2}=0$, and by the same argument, $W_{1,j} = W_{j,1} = 0$ for all $j$.

(ii) Since all diagonal blocks $W_{i,i}\ge 0$, their $k$-th rows and columns vanish whenever the $(k,k)$-entry is zero. If any $W_{i,j}$ had a nonzero entry in these rows or columns, we could again construct (as in (i)) a vector $\ket{b} \in \bbC^m$ making $(\bra{b}\otimes I_n)W(\ket{b}\otimes I_n)$ non-positive, violating the block-positivity of $W$.
\end{proof}

{\bf Remark.} An equivalent form of the above lemma is that $\cR(W)\subseteq \cR(W_{A})\otimes \cR(W_{B})$, where $W_A := [\tr(W_{i,j})]_{i,j=1}^m$ and $W_B := \sum_{i=1}^m W_{i,i}$. 

The following lemma establishes an upper bound for the number of negative eigenvalues of block-positive matrices using methods from algebraic geometry.
\begin{lemma}\cite{sarbicki2008spectral}
	\label{pv}
	(i)	Suppose $\cV\subseteq \bbC^m\otimes \bbC^n$ with $\dim(\cV)>(m-1)(n-1)$, then $\cV$ contains at least one product vector.
	
	(ii) Suppose $W\in \bp_{m,n}$, then W has no more than $(m-1)(n-1)$ negative eigenvalues.
\end{lemma}

Given a matrix belonging to $\bp_{m,n}$, if it has at least one negative eigenvalue, then it is a legitimate entanglement witness (EW).
For convenience, we refer to such EW as an $m\times n$ EW.
A decomposable entanglement witness (DEW) admits the expression
$P+Q^\G$ with $P,Q\ge 0$, otherwise, it is called nondecomposable EW (NDEW). By definition, DEW cannot detect PPT entangled states, i.e, $\tr(W \sigma)\ge 0$ for all PPT states $\sigma$. Further, $W^\G$ is also an NDEW, given that $W$ is an NDEW, since it remains block-positive and satisfies $\tr(W^\G\r^\G)=\tr(W\r)<0$, where $\r$ is a PPT entangled state detected by $W$. For $mn\le 6$, since every $m\times n$ PPT state is separable, we have every $m\times n$ EW is decomposable. For higher-dimensional systems, NDEWs have to be constructed to detect PPTES. One of such way is based on edge states.
\begin{lemma}\cite{lewenstein2001characterization}
	\label{le:ndew-edge}
	Any $m\times n$ NDEW can be represented as $P+Q^\G-\epsilon I_{mn}$, with  $0<\e\le \inf_{\ket{c,d},\norm{c}=\norm{d}=1}\bra{c,d}(P+Q^\G)\ket{c,d}$. Moreover,
	$P, Q\ge 0$ such that $\tr(P \delta)=\tr(Q^\G \delta)=0$ for some edge state $\delta$.
\end{lemma}

One generic form is that $P$ and $Q$ denote the projectors onto the kernel space of $\sigma$ and $\sigma^\G$, respectively. The positivity of $\epsilon$ follows from the edge state properties. However, we note that not every pair of $P$ and $Q$ satisfying
$\tr(P \delta)=\tr(Q^\G \delta)=0$ can form a legitimate NDEW. For example, let $m=n=3$ and $\{\ket{a_j,b_j}\}_{j=1}^5$ be a set of real two-qutrit unextendible product basis (such basis is a set of orthogonal product vectors that contain no product vector in their complementary space) \cite{divincenzo2003unextendible}. Then $\s=I_9-\sum^5_{j=1}\proj{a_j,b_j}$
is a two-qutrit edge state of rank four. Let $P=Q=Q^\G=\sum^4_{j=1}\proj{a_j,b_j}$. We have $\tr(P \s)=\tr(Q^\G \s)=0$. However, 
the infimum of $\epsilon$ is zero as we can choose $\ket{c,d}=\ket{a_5,b_5}$. This contradicts the fact that $\epsilon$ is strictly positive.

Among EWs, optimal witnesses are of particular importance. 
The knowledge of all optimal entanglement
witnesses fully characterizes the set of entangled states.
Given an $m\times n$ EW $W$, we refer to $D_W$ as the set of all entangled states detected by $W$. We say that $W_1$ is finer than $W_2$ when $D_{W_1} \supseteq D_{W_1}$. Next, $W$ is optimal when no EW is finer than $W$. It is proved that $W$ is an optimal EW iff for any nonzero positive operator $P$, the operator $W-P$ is no longer block-positive \cite{lewenstein2000optimization1}. By definition, if an EW $W$ satisfies that $\tr(W\r)=0$ for a full-rank separable state $\r$, then $W$ is optimal.
Let $P_W$ consist of product vectors $\ket{a,b}$ such that $\bra{a,b}W\ket{a,b}=0$. We say that $W$ has the spanning property if $\span P_W=\bbC^m\otimes\bbC^n$. It is known that any EW possessing the spanning property is optimal \cite{lewenstein2000optimization1}. 
Further, a linear bipartite subspace $\cS$ is called completely entangled (CES) if there is no product vector in $\cS$. For DEW, it is straightforward to verify the following claim holds.
\begin{lemma}
	\label{opdc}
	A DEW $W$ is optimal if and only if $W=Q^\G$, where $Q\ge 0$ is supported on some CES.
\end{lemma}

An EW is called extremal if $W-B$ is no longer an EW, where $B$ is an arbitrary block-positive operator such that $B\neq k W$. By definition, any extremal EW is optimal. For DEWs, $W$ is extremal if and only if it can be written as $\proj{\psi}^\G$ for some entangled state $\ket{\psi}$ \cite{marciniak2010extremal}.

Two matrices $A,B\in \cM_m(\bbC)\otimes \cM_n(\bbC)$ are called locally equivalent if there exist invertible matrices $X\in \cM_m(\bbC)$ and 
$Y\in \cM_n(\bbC)$ such that $(X \otimes Y) A (X \otimes Y)^\dg =B$. 
By Sylvester’s Theorem, two locally equivalent matrices have the same inertia. Hence, it is straightforward to verify that all the above-mentioned notions and properties are invariant up to local equivalence. For example, 
let $A,B$ be invertible matrices, then $(A \otimes B)\r(A \otimes B)^\dg$ remains a PPT (NPT) state  given that $\r$ is a PPT (NPT) state, $(A \otimes B)W(A \otimes B)^\dg$ remains an optimal EW given that $W$ is an optimal EW, and so on.

\subsection{Absolutely separable and PPT states}
According to Lemma \ref{ineqe} (iv), the spectral properties of an entanglement witness (EW) can be characterized using results pertaining to absolutely separable and positive partial transpose (PPT) states, as will be discussed in the following section. An absolutely separable state (AS) is a state that remains separable under any global unitary transformation \cite{kus2001geometry, knill2003separability}.  A first-coming example of such states is the maximally mixed state.
We also know from the definition that absolute separability is a spectral property, and the problem is to find conditions on the spectrum of a state for it to be AS. One motivation for characterizing such states is that it is experimentally easier to determine the eigenvalues of a state rather than reconstructing the state itself \cite{ekert2002direct,tanaka2014determining}. Analogous to AS states, states that remain PPT under any global unitary transformation are termed absolutely PPT (AP) states. In the ensuing discussion, we will denote the sets of AS states and AP states in 
$\cM_m(\bbC)\otimes \cM_n(\bbC)$ as $\as_{m,n}$ and $\app_{m,n}$, respectively. It directly follows that $\as_{m,n}\subseteq \app_{m,n}$. The characterization of $\as_{2,2}$ was initially provided in \cite{verstraete2001maximally}. 
It was also proved that there is a ball of AS states centered at the maximally mixed state, which is known as the {\em maximal ball}. Specifically,
if the $m\times n$ state $\r$ satisfies $\tr(\rho^2)\le \frac{1}{mn-1}$, then $\rho\in \as_{m,n}$ \cite{gurvits2002largest}. Building on this result, consider the family of $m\times m$ states
\begin{eqnarray}
	\label{zzz1}
	\z_1:=c\diag(\frac{m+1}{m-1},\cdots,\frac{m+1}{m-1},1,\cdots,1),
	\end{eqnarray}
where $c$ is the normalization factor, and the multiplicity of eigenvalue $\frac{m+1}{m-1}$ is $l$. It can be verified by direct computation that
 $\tr(\z_1^2)\le \frac{1}{m^2-1}$ holds for any $l\in [1,m^2]$. Hence $\z_1\in \as_{m,m}$.
On another front, by characterizing the random robustness of pure states, it was proved that the unnormalized state $\proj{\psi}+\frac{1}{2}I_{mn}$ is separable for arbitrary pure state $\ket{\psi}\in \bbC^m\otimes \bbC^n$ \cite{Vidal1999Robustness}, which implies that 
\begin{eqnarray}
	\label{311}
	\z_2:=\frac{1}{mn+2}\diag(3,1,\cdots,1)\in \as_{m,n}.
\end{eqnarray}

Regarding the set $\app_{m,n}$, a necessary and sufficient condition for a state to belong to it has been established, represented by a finite set of linear matrix inequalities \cite{hildebrand2007positive}. 
Later, it was proved that $\as_{2,n}=\app_{2,n}$ for arbitrary $n$ \cite{johnston2013separability}. The criterion indicates that $\r\in \as_{2,n}$ ($\app_{2,n}$) if and only if 
\begin{eqnarray}
	\label{as2n}
\lambda_1\le \lambda_{2n-1}+2\sqrt{\lambda_{2n-2}\lambda_{2n}}.
\end{eqnarray}
However, the problem of whether the two sets $\as_{m,n}$ and $\app_{m,n}$ are identical for $m,n\ge 3 $ still remains open. 
Note that the number of linear inequalities needed to determine whether
a state belongs to $\app_{m,n}$ grows exponentially with the dimension $\min\{m,n\}$. Hence, in general, it is a difficult task to determine an AP state in higher-dimensional systems. In the following lemma, we propose two types of AP states based on matrix analysis techniques.

\begin{lemma}
	\label{twoap}
	For any $m,n\ge 2$, the (unnormalized) states
	\begin{eqnarray}
	&&\r_1=\diag(\sqrt{2}+1,\sqrt{2}+1,1,\cdots,1)\in \app_{m,n},\\
	&&\r_2=\diag(2,2,2,1,\cdots,1)\in \app_{m,n}.
	\end{eqnarray}
\end{lemma}
\begin{proof}
It is known from \cite{hildebrand2007positive} that $\app_{m,n}=\app_{n,m}$. Hence we can assume that $m\le n$. To prove that $\r_1\in \app_{m,n}$, it suffices to prove that $\tr((U\r_1 U^\dg)^\G\cdot\proj{x})=\tr((U\r_1 U^\dg)\cdot\proj{x}^\G)\ge 0$ holds for any global unitary matrix $U$ and unit vector
$\ket{x}$. Denote the Schmidt coefficients of $\ket{x}$ as $x_1\ge \cdots\ge x_m$. Using Lemma \ref{ineqe} (iv), we have
\begin{eqnarray}
	\label{kk1}
	\notag
&&\tr((U\r_1 U^\dg)\cdot\proj{x}^\G)\ge \sum_{j=1}^{mn}\l_j(\r_1)\cdot\l_{mn+1-j}(\proj{x}^\G)\\
=&&(\frac{\sqrt{2}}{2}x_1-x_2)^2+(\frac{\sqrt{2}}{2}x_1-x_3)^2+x_4^2+\cdots+x_m^2\ge 0, 
\end{eqnarray}
where the equality follows from Lemma \ref{ept} that
$\l_{mn}(\proj{x}^\G)=-x_1 x_2$ and $\l_{mn-1}(\proj{x}^\G)=-x_1 x_3$. 
This proves that $\r_1\in \app_{m,n}$.

To continue, for a given $\ket{x}$, the value of $\l_{mn-2}(\proj{x}^\G)$ can be either $-x_2x_3$ or $-x_1x_4$. If $\l_{mn-2}(\proj{x}^\G)=-x_2x_3$, then
\begin{eqnarray}
	\label{zym1}
	\notag
	&&\sum_{j=1}^{mn}\l_j(\r_2)\cdot\l_{mn+1-j}(\proj{x}^\G)\\=&&(\frac{\sqrt{2}}{2}x_1-\frac{\sqrt{2}}{2}x_2)^2+(\frac{\sqrt{2}}{2}x_1-\frac{\sqrt{2}}{2}x_3)^2+(\frac{\sqrt{2}}{2}x_2-\frac{\sqrt{2}}{2}x_3)^2
	+x_4^2+\cdots+x_m^2\ge 0.
\end{eqnarray}
If $\l_{mn-2}(\proj{x}^\G)=-x_1x_4$, then
\begin{eqnarray}
	\label{zym2}
	\sum_{j=1}^{mn}\l_j(\r_2)\cdot
 \l_{mn+1-j}(\proj{x}^\G)&&=1-x_1(x_2+x_3+x_4)\ge 1-x_1\sqrt{3(1-x_1^2)}\ge 1-\frac{\sqrt{3}}{2}>0.
\end{eqnarray}
We conclude from (\ref{zym1}) and (\ref{zym2}) that $\r_2\in \app_{m,n}$.
\end{proof}

Notice that it remains unclear so far whether the above two states belong to $\as_{m,n}$. In the following, we shall use the above lemma to obtain two spectral properties of DEWs.

\section{Spectral properties of entanglement witnesses}
\label{sec:res=eigenvalues}
In this section, we characterize several spectral properties of EWs. 
We will examine both DEWs and NDEWs separately. First, we establish several fundamental properties common to both categories.

\begin{lemma}
	\label{use}
	Let $W$ be an arbitrary normalized $m\times n$ EW. Then
	
    (i) $\tr(W^2)\in (\frac{1}{mn-1},1]$.

	(ii) $\l_{mn}\in [-\frac{1}{2},0)$. If $\l_{mn}=-\frac{1}{2}$, then W is optimal. 
	
	(iii)  $\l_1\in (\frac{1}{mn-1},1)$.

	(iv) For $m=2$: (a) $\l_2+\l_{2n}\ge 0$, (b) $\sum_{i=3}^{2n} \l_i\ge -\frac{1}{2+2\sqrt{2}}$ and (c)
	$\sum_{i=k}^{2n} \l_{i}\ge -\frac{1}{2}$ for any $4\le k\le 2n$.
	
	(v) For $m=n$, $\cN(W)\le \frac{m-1}{2}$. If the inequality is saturated, then $W$ is optimal and has eigenvalues $\frac{1}{m}$ with multiplicity $\frac{m(m+1)}{2}$, and $-\frac{1}{m}$ with multiplicity $\frac{m(m-1)}{2}$. 
	\end{lemma}
\begin{proof}
 (i) From \cite{szarek2008geometry}, any block-positive matrix $W$ satisfies that
    \begin{eqnarray}
    	\label{xp}
    \tr(W^2)\le	\tr(W)^2.
    \end{eqnarray}
 So we have $\tr(W^2)\le 1$. Further, since $\l_{mn}<0$ by definition of EWs, we have
   $\tr(W^2)\ge \frac{1}{mn-1}(1-\l_{mn})^2+\l_{mn}^2>\frac{1}{mn-1}$,
where the first inequality follows from the Cauchy-Schwarz inequality. Thus the claim holds.

	(ii) The inequality $\l_{mn}<0$ holds by definition. Let $U$ be the order-$mn$ unitary matrix such that 
	\begin{eqnarray}
		\label{UDIAG}
	UWU^\dg=\diag(\l_{mn},\cdots,\l_1).
	\end{eqnarray}
	Since $\z_2\in \as_{m,n}$ from (\ref{311}), we have
	$\tr(UWU^\dg\cdot\z_2)=\frac{1}{mn+2}(2\l_{mn}+1)=\tr(W\cdot U^\dg\z_2 U)\ge 0$, as $U^\dg \z_2 U$ is separable. This implies that $\l_{mn}\ge -\frac{1}{2}$. Suppose $\l_{mn}=-\frac{1}{2}$.
	Then we have $\tr(W\cdot U^\dg\z_2 U)=0$, where $U^\dg \z_2 U$ is a full-rank separable state. This implies that $W$ possesses the spanning property, and thus is optimal.

(iii) If $\l_1\le \frac{1}{mn-1}$, then the normalized condition implies that $\l_{mn}\ge 0$, which contradicts with the definition of EWs. The inequality $\l_1<1$ directly follows from (\ref{xp}).

(iv) We first prove (a). For the two-qubit case, the result follows from \cite[Theorem 3]{johnston2018inverse}. For $2\times n$ ($n>2$) systems, assume that 
 there exists a EW $W$ such that $\l_2(W)+\l_{2n}(W)<0$. Up to local unitary equivalence, we can assume the eigenvector corresponding to $\l_{2n}(W)$ is $\cos t \ket{11}+\sin t \ket{22}$ with $t\in(0,\pi/2)$. 
 Projecting $W$ onto the subspace spanned by $\{\ket{1},\ket{2}\}\otimes \{\ket{1},\ket{2}\}$, we obtain a $2\times 2$ block-positive matrix $W'$. One can verify that $\l_4(W')=\l_{2n}(W)<0$. Thus $W'$ is a two-qubit EW.
Further, using Lemma \ref{ineqe} (v), we have $\l_2(W')\le \l_2(W)$. This implies that $\l_2(W')+\l_4(W')<0$, which leads to a contradiction. Thus the claim holds.

We next prove (b). According to (\ref{as2n}), the (unnormalized) state $\g:=\diag(3+2\sqrt{2},\cdots,3+2\sqrt{2},1,1)\in \as_{2,n}$. 
Combining with the expression (\ref{UDIAG}), we have $\tr(UWU^\dg\cdot\g)\ge 0$, from which the desired result follows. 

Finally, the proof of (c) is similar to that of (b), noting that the (unnormalized) state $\diag(3,\cdots3,1,\cdots,1)\in \as_{2,n}$, where the multiplicity of eigenvalue 3 is $2n-k+1$, again by (\ref{as2n}).

(v) The proof is similar to that of (ii). Suppose $W$ has $m^2-k$ negative eigenvalues, where $2m-1\le k\le m^2-1$ follows from Lemma \ref{pv} (ii). By the expression of (\ref{UDIAG}), we have 
\begin{eqnarray}
	\label{fr}
	\tr(UWU^\dg \cdot\z_1)=1+\frac{2}{m-1}(\l_{m^2}+\cdots+\l_{k+1})=\tr(W\cdot U^\dg \z_1 U)\ge 0,
\end{eqnarray}
where $\z_1$ in (\ref{zzz1}) is absolutely separable for $l=m^2-k$.

Let $\cN(W)=\frac{m-1}{2}$. Then from (\ref{fr}) we have $\tr(W\cdot U^\dg \z_1 U)=0$, where $U^\dg \z_1 U$ is a full-rank separable state. This already implies that $W$ is optimal. On the other hand, we have $\sum_{i=1}^{k}\l_i=\frac{m+1}{2}$ and $\sum_{i=k+1}^{m^2}\l_i=\frac{1-m}{2}$.
Consequently,
\begin{eqnarray}
	\label{kkk}
	\tr(W^2)=\sum_{i=1}^{k}\l_i^2+\sum_{i=k+1}^{m^2}\l_i^2
	\ge\frac{(m+1)^2}{4k}+\frac{(m-1)^2}{4(m^2-k)},
\end{eqnarray}
where the inequality follows from the Cauchy-Schwarz inequality. By viewing $\tr(W^2)$ as a function of $k$ and computing its derivative function, one can verify that $\tr(W^2)$ attains mimimum when $k=\frac{m(m+1)}{2}$. This implies that $\tr(W^2)\ge 1$. Recalling from (i) that $\tr(W^2)\le 1$. We have $\tr(W^2)=1$. Hence $k=\frac{m(m+1)}{2}$, and 
the saturation condition of (\ref{kkk})
implies the desired results of eigenvalues. 
\end{proof}

{\bf Remark.} 
\begin{enumerate}
	\item The inequality (a) in (iv) does not hold for $m\ge 3$, even for the weaker inequality $\l_1+\l_{mn}\ge 0$.
	For example, consider the family of two-qutrit generalized Choi witness $W[0,1,1]$ \cite[(7.21)]{Chru2014Entanglement}, which has eigenvalues $(-2,1,1,1,1,1,1,1,1)$.
	
	\item 
	 In the following, we shall see that the inequality $\cN(W)\le \frac{m-1}{2}$ in (v) holds for $m\times n$ normalized DEWs where $m<n$. However, it is unclear whether this inequality holds for $m\times n$ NDEWs with $m<n$, as $\z_1$ in (\ref{zzz1}) is known to be absolutely separable only when $m=n$. 
\end{enumerate}

It is known that the set of deficient-rank entangled states encompasses many important families of entangled states studied over the past decades, such as 
the qubit-qudit and two-qutrit edge states \cite{horodecki1997separability}, extreme PPT entangled states \cite{chen2013properties}, as well as $n$-copy-undistillable two-qutrit NPT states of rank five \cite{chen2016distillability1}. In the following, we shall use these states to 
investigate some spectral properties of EWs.

\begin{lemma}
	\label{use2}
Any deficient-rank NPT (resp. PPT) entangled state can be detected by a normalized DEW (resp. NDEW) which is arbitrarily close to a pure state. 
Consequently, normalized $W$ exists for both DEWs and NDEWs such that $\l(W)$ is arbitrarily close to $(1,0,\cdots,0)$.
\end{lemma}
\begin{proof}
Given a deficient-rank NPT (resp. PPT) entangled state $\r$, there exists a normalized DEW (resp. NDEW) $W$ detecting it. Since $\r$ does not have full rank, there exists a pure state $\ket{\psi}$ such that $\bra{\psi}\r\ket{\psi}=0$. Let $W_1=(1-p)W+p\proj{\ps}$ with $p\in(0,1)$. We have $W_1$ is also a DEW (resp. NDEW) since it remains block-positive and detects the entanglement of $\r$. Thus the claim holds by letting $p\ra1$.
Consequently, by Lemma \ref{ineqe} (iii), we have $\l(W_1)$ can be arbitrarily close to $\l(\proj{\psi})=(1,0,\cdots,0)$.
\end{proof}

We now state our first main theorem that focuses on DEWs only.
The proof is contained in Appendix \ref{proof:pt1012}.

\begin{theorem}
	\label{th:lambdaIN[-1/2,1)}

Let $W$ represent an element of $m\times n$ normalized DEWs ($m\le n$). 

(i) For each $x\in (\frac{1}{mn-1},1]$, there exists a $W$ such that $\tr(W^2)=x$. The infimum of $\tr(W^2)$ is $\frac{1}{mn-1}$ and not attainable. The supremum of $\tr(W^2)$ is 1 and attainable.

(ii) $\tr(W^2)=1$ if and only if $W$ is the partial transpose of a pure entangled state.

(iii) For each $x\in [-\frac{1}{2},0)$, there exists a $W$ such that $\l_{mn}(W)=x$. The infimum of $\l_{mn}$ is $-\frac{1}{2}$ and 
attainable. The supremum of $\l_{mn}$ is $0$ and not
attainable.

(iv)
$\l_{mn}(W)=-\frac{1}{2}$ if and only if $W$ is the partial transpose of a pure state whose two nonzero Schmidt coefficients are both $\frac{\sqrt{2}}{2}$. 
	
(v) For each $x\in (\frac{1}{mn-1},1)$, there exists a $W$ such that $\l_1(W)=x$. The infimum of $\l_1(W)$ is $\frac{1}{mn-1}$ and not attainable. The supremum of $\l_1(W)$ is 1 and not attainable.

(vi) $0<\cN(W)\le \frac{m-1}{2}$. For each $x\in (0,\frac{m-1}{2}]$, there exists a $W$ such that $\cN(W)=x$. The supremum of $\cN(W)$ is $\frac{m-1}{2}$ and attainable. The infimum of $\cN(W)$ is 0 and not attainable.

(vii) For the integer $1\le j\le \lfloor\frac{n}{m}\rfloor$, we define the maximally entangled state $\ket{\Phi_j}:=\frac{1}{\sqrt{m}}\sum_{i=1}^m \ket{i,i+(j-1)m}\in \bbC^m\otimes \bbC^n$. Then $\cN(W)=\frac{m-1}{2}$ if and only if under local unitary equivalence, $W$ can be written as the convex combination of $\proj{\Phi_1}^\G,\cdots,\proj{\Phi_k}^\G$, where $1\le k\le \lfloor \frac{n}{m} \rfloor$.

(viii) Let $m=2$. The infimum of $\sum_{i=3}^{2n} \l_i$ is $-\frac{1}{2+2\sqrt{2}}$, which is attained by $(x_1\ket{11}+x_2\ket{22})
(x_1\bra{11}+x_2\bra{22})^\G$, where $x_1^2+x_1x_2=1-\frac{1}{2+2\sqrt{2}}$ (numerically, $x_1=0.92388, x_2=0.382683$). For any $4\le k\le 2n$, the infimum of 
$\sum_{i=k}^{2n} \l_{i}$ is $-\frac{1}{2}$, which can be attained by the partial transpose of a qubit-qudit maximally entangled state.

(ix) For $3\le m\le n$, the infima of $\l_{mn}+\l_{mn-1}$ and  $\l_{mn}+\l_{mn-1}+\l_{mn-2}$
are $-\frac{\sqrt{2}}{2}$ and $-1$ respectively. The first infimum can be attained by the partial transpose of a pure state whose nonzero Schmidt coefficients are $\frac{\sqrt{2}}{2},\frac{1}{2},\frac{1}{2}$.  The second can be attained by the partial transpose of a two-qutrit maximally entangled state. 
\end{theorem}

Our next theorem characterizes the similar properties for NDEWs with the proof given in Appendix \ref{ndd}. 
\begin{theorem}
	\label{NDEW}
	Let $W$ represent an element of $m\times n$ normalized  NDEWs ($m\le n$). Then
	
	(i) The supremum of $\tr(W^2)$ is 1.
	
	(ii) The infimum of $\tr(W^2)$ is not attainable.
	
    (iii) The supremum of $\l_{mn}$ is 0 and  not attainable.

    (iv) The infimum of $\l_{mn}$ is $-\frac{1}{2}$ and not attainable.

   (v) The supremum of $\l_1$ is 1 and not attainable.

	(vi) The infimum of $\l_1$ is not attainable.
	
	(vii) For $m=2$ and $m=n\ge 3$, the supremum of $\cN(W)$ is $\frac{m-1}{2}$.  
	
	(viii) For $m=n=3$, the supremum of $\cN(W)$ is not attainable.
	
	(ix) The infimum of $\cN(W)$ is $0$ and not attainable.
	\end{theorem}

We have seen from Theorem \ref{th:lambdaIN[-1/2,1)} (iii) and Theorem \ref{NDEW} (iv) that the infimum of $\l_{mn}$ can be attained by DEWs but not by NDEWs, which reveals a difference between their spectral properties. 
Nevertheless, in comparison to DEWs, the spectral properties of NDEWs remain less understood. In particular, while the infimums of $\lambda_1$ and $\tr(W^2)$ remain unknown, we observe that the infimum of $\lambda_1$ is $\frac{1}{mn-1}$ if and only if the infimum of $\tr(W^2)$ equals $\frac{1}{mn-1}$. Moreover, the supremum of $\mathcal{N}(W)$ for $m<n$ and whether it is attainable is still an open question. We prove in (viii) that the supremum of $\cN(W)$ for $m=n=3$ is not attainable. We also conjecture that this holds for arbitrary $m=n$. See the following conjecture.

\bcj 
\label{cnnb}
The supremum $\frac{m-1}{2}$ of $\cN(W)$ for normalized $m\times m$ NDEWs is not attainable.
\ecj
From Lemma \ref{use} (v), we have known that if $\cN(W)=\frac{m-1}{2}$, then $W$ is optimal and can be written as $W=\frac{1}{m}(I_{m^2}-P)$ where $P$ is the rank-$\frac{m(m-1)}{2}$ projection onto the negative eigenspace of $W$.The block-positivity of $W$ requires that 
$\bra{x,y}P\ket{x,y}\le \frac{1}{2}$ for any unit vectors $\ket{x},\ket{y}\in \bbC^m$. The optimality of $W$ requires that there exists unit vectors $\ket{x},\ket{y}\in \bbC^m$ such that
$\bra{x,y}P\ket{x,y}=\frac{1}{2}$. For $m=3$, we have demonstrated in the proof of Theorem \ref{NDEW} (viii) that $P$ is spanned by the three basis vectors of antisymmetric subspace in $\bbC^3\otimes \bbC^3$, under local unitary equivalence. We conjecture that this property extends to arbitrary dimensions.

\bcj
\label{anti}
If the maximal Schmidt coefficient of all unit vectors in a $\frac{m(m-1)}{2}$-dimensional subspace of $\bbC^m\otimes \bbC^m$ is at most $\frac{\sqrt{2}}{2}$, then the subspace is spanned by the antisymmetric basis under local unitary equivalence.
\ecj
Previous results have shown that the maximal Schmidt coefficient of any vector in antisymmetric space is at most $\frac{\sqrt{2}}{2}$ \cite[Proposition 1]{grudka2010constructive}. But as far as we know, the converse remains unproven. If this conjecture is true, then Conjecture \ref{cnnb} would also follow, since in that case one could directly conclude that $W^\G\ge 0$ when $\cN(W)=\frac{m-1}{2}$.

\section{Mirrored pair of EWs and Detecting Power of NDEWs}
\label{Power}
The detection capability of entanglement witnesses (EWs), which refers to the number and types of entangled states they can identify, is a crucial aspect to investigate. Previous results have indicated a fundamental limitation of the detectability of an EW without prior knowledge of the target states \cite{liu2022fundamental}. In general, for a given EW, its detection capability can be enhanced through optimization techniques, which essentially involve subtracting appropriate positive operators \cite{lewenstein2000optimization1}.
An alternative approach to enhance detection capability is provided by the framework of mirrored EWs \cite{bae2020mirrored}.
Given an $m\times n$ EW $W$, one defines a mirrored operator $W_M=\mu I_{mn}-W$, where $\mu$ is the smallest number such that $W_M$ is block-positive. Equivalently, we have $\mu=\sup_{\ket{a,b},\norm{a}=\norm{b}=1}\bra{a,b}W\ket{a,b}$.
If  $\l_{1}(W)>\mu$, then $W_M$ is a legitimate EW and $(W,W_M)$ constitutes a pair of EWs, which can double up the capability of detecting entangled states. Several examples of such mirrored pairs of EWs involving both DEWs and NDEWs have been demonstrated \cite{bae2020mirrored,bera2023structure}. Recently, a mirrored pair of NDEWs $(W,W_M)$ has been constructed, where $W$ and $W_M$ are both optimal \cite{chruscinski2025mirroredn}.
However, not every EW has a mirrored EW \cite{bera2023structure}. Consider an $m\times n$ extremal DEW $W=\proj{\psi}^\G$ for some pure entangled state $\ket{\psi}\in \bbC^m\otimes \bbC^n$ with Schmidt coefficients $s_1\ge \cdots\ge s_m\ge 0$. In this case, by using Lemma \ref{ept}, we have
\begin{eqnarray}
	\label{maxdef}
\mu=\sup_{\ket{a,b},\norm{a}=\norm{b}=1}\bra{a,b}(\proj{\psi}^\G)\ket{a,b}=s_1^2.
\end{eqnarray}
Consequently, $W_M=s_1^2 I_{mn}-\proj{\psi}^\G\ge 0$ cannot be an EW. This leads to the question of determining the conditions for a mirrored pair of EWs. Here we apply our established results from the previous section to provide two necessary conditions.

\begin{corollary}
	Given an $m\times n$ normalized EW $W$ ($m\le n$),
	a necessary condition for the mirrored operator $W_M$ to be an EW
	is that $\l_{mn}(W)>-\frac{1}{2}$. If $W$ is decomposable, then two additionally necessary conditions are $\tr(W^2)<1$ and $\cN(W)<\frac{m-1}{2}$.
	\end{corollary}
\begin{proof}
Recall from Lemma \ref{use} that $\l_{mn}(W)\ge -\frac{1}{2}$.
Assume that $\l_{mn}(W)=-\frac{1}{2}$. Following from Theorem \ref{NDEW} (iv), $W$  must be decomposable. Consequently, following Theorem \ref{th:lambdaIN[-1/2,1)} (iv), we obtain that $W$ must be the partial transpose of a pure entangled state with two nonzero Schmit coefficients $\frac{\sqrt{2}}{2}$. Thus, $W_M$ cannot be an EW. 

We next prove the second claim. For any normalized DEWs, 
Lemma \ref{use} (i) and Theorem \ref{th:lambdaIN[-1/2,1)} (vi) indicate that $\tr(W^2)\le 1$ and $\cN(W)\le \frac{1}{2}$.
Firstly, suppose $W$ satisfies $\tr(W^2)=1$. Then from Theorem \ref{th:lambdaIN[-1/2,1)} (ii), it must be the partial transpose of a pure entangled state, which implies that $W_M$ cannot be an EW. Next, suppose $W$ is an $m\times n$ DEW such that $\cN(W)=\frac{m-1}{2}$. Using Theorem \ref{th:lambdaIN[-1/2,1)} (vii), $W$ can be written as the convex combination $W=\sum_{i=1}^k p_i\proj{\Phi_i}^\G$ under local unitary equivalence. Without loss of generality, we assume that $\max\limits_i p_i=p_1$.
Given any unit vector $\ket{a,b}\in \bbC^m\otimes \bbC^n$, by the expression of $\ket{\Phi_i}$, we have 
\begin{eqnarray}
	\bra{a,b}W\ket{a,b}=\sum_{i=1}^k p_i \bra{a,b_i}(\proj{\Phi_i}^\G)\ket{a,b_i},
\end{eqnarray}
where $\ket{b_i}\in \bbC^n$ denotes the vector whose entries from $1+(i-1)m$ to $m+(i-1)m$ coincide with those of $\ket{b}$, and all other entries are zero. It then follows from (\ref{maxdef}) that
\begin{eqnarray}
	\bra{a,b}W\ket{a,b}\le \sum_{i=1}^k p_i \frac{1}{m}||\ket{a,b_i}||^2 \le \frac{p_1}{m} (\sum_{i=1}^k ||\ket{a,b_i}||^2)=\frac{p_1}{m},
\end{eqnarray}
where the upper bound can be achieved by choosing $\ket{a,b}=\ket{1,1}$. We have $W_M=(\frac{p_1}{m})I_{mn}-W$ and thus $W_M\ge 0$ cannot be an EW. This completes the proof.
\end{proof}
{\bf Remark.}
The above conditions are not sufficient. For example, consider the two-qubit normalized EW $W=\frac{1}{3}(\ket{11}+\ket{22})(\bra{11}+\bra{22})^\G+\frac{1}{3}\proj{11}$,
which satisfies that $\l_4(W)>-\frac{1}{2}$ and $\cN(W)<\frac{1}{2}$. Consequently, $\mu=\l_1(W)=\frac{2}{3}$. Therefore, $W_M=\frac{2}{3}I_4-W\ge 0$ is not an EW.

To conclude this section, we examine the detection capability of NDEWs.
Working with the Hilbert-Schmidt Euclidean measure, previous work \cite{szarek2008geometry} showed that as the dimension of the systems increases, the set of DEWs constitutes a vanishingly small subset among all
EWs. While every NPT state can be detected by DEWs, it is a natural question whether every $m\times n$ NPT state can be detected by NDEWs for $mn>6$. Our next theorem answers this question affirmatively. The complete proof appears in Appendix \ref{adem}.

\begin{theorem}
	\label{NPTNDEW}
	Any $m\times n$ entangled state with $mn>6$ can be detected by an NDEW.
\end{theorem}

This result is consistent with the above geometric measure of EWs, since most of EWs are NDEWs in higher-dimensional systems. A subsequent direction would be to identify specific subclasses of NDEWs capable of detecting all entangled states.

\section{Conclusions}
\label{sec:con}
We have provided a spectral analysis of entanglement witnesses, characterizing the suprema and infima of 
largest eigenvalue, smallest eigenvalue, negativity and squared Frobenius norm of a normalized EWs. For DEWs, the characterization is complete. Specifically, all infima and suprema are explicitly determined, and we have provided necessary and sufficient conditions for their attainability. 
As an application of our results, we have derived some necessary conditions for an EW to have a mirrored counterpart.   We have also demonstrated that any  entangled states beyond two-qubit and qutrit systems can be detected by NDEWs, implying a stronger detection capability of NDEWs.
However, due to their more complex structure, the corresponding spectral analysis for NDEWs remains less comprehensive, leaving several open questions. For a normalized $m\times n$ NDEW, these include determining: (a) the infima of $\l_1(W)$ and $\tr(W^2)$, (b) the supremum of $\cN(W)$ when $m<n$, and (c) the attainment of suprema of $\tr(W^2)$ and $\cN(W)$.  For a special case that $m=n$, we conjecture the suprema of $\cN(W)$ is not attaineable. Further investigation of Conjecture \ref{anti} is required to explore this question.

\section*{ACKNOWLEDGMENTS}
Authors were supported by the NNSF of China (Grant No. 12471427).

\appendix
\section{Proof of Theorem \ref{th:lambdaIN[-1/2,1)}}
\label{proof:pt1012}

Let $\ket{\Omega}={1\over\sqrt2}(\ket{12}-\ket{21})$ and  $\ket{\Psi_m}=\frac{1}{\sqrt{m}}\sum_{i=1}^m \ket{ii}\in \bbC^m\otimes \bbC^n$.
Define
\begin{eqnarray}
	\label{abcd}
	W_{a,b,c,d}:=
	a\frac{1}{mn-1}(I_{mn}-\proj{\Omega})+b\proj{\Psi_2}^\G
	+c \proj{11}+d \proj{\Psi_m}^\G,
\end{eqnarray}
where $a,b,c,d\in [0,1]$ and $a+b+c+d=1$. We have $W_{a,b,c,d}\in \bp_{m,n}$.

(i) A direct computation shows that the eigenvalues of $W_{a,b,0,0}$ are ${1-b\over mn-1}$ (multiplicity $mn-4$), ${1-b\over mn-1}+{b\over2}$ (multiplicity three)  and $-\frac{b}{2}$ (multiplicity one). We have $\l_{mn}(W_{a,b,0,0})=-\frac{b}{2}$, and thus $W_{a,b,0,0}$ is a legitimate DEW whenever $b\in (0,1]$. Further, one can verify by computing the derivative that $\tr W_{a,b,0,0}^2$ monotonically increases with $b$. This implies that $\tr (W^2)\in({1\over mn-1},1]$, and for every $x$ in this interval, such a $W$ exists. Combining with Lemma \ref{use} (i), the infimum of $\tr(W^2)$ is $\frac{1}{mn-1}$ and 
not attainable. The supremum of $\tr(W^2)$ is 1 and can be attained by the normalized DEW $W_{0,1,0,0}$.

(ii) The if part can be verified straightforwardly. We prove the only if part.
Suppose a normalized DEW $W$ satisfies $\tr(W^2)=1$. Write $W=x P+(1-x)Q^\G$, where $x\in [0,1)$, $P,Q\ge 0$ and $\tr(P)=\tr(Q^\G)=\tr(Q)=1$. We have $\tr(P^2)\le 1$ and $\tr(Q^2)=\tr(Q^\G Q^\G)\le 1$.
Consequently,
\begin{eqnarray}
	\label{eq:trW^2}
	\notag
	\tr W^2&&=x^2 \tr(P^2)+(1-x)^2\tr(Q^2)+2x(1-x)\tr(P Q^\G)\\
	\notag
	&&\le x^2+(1-x)^2+2x(1-x)\sqrt{\tr(P^2)}\sqrt{\tr(Q^2)}\\
	&&\le x^2+(1-x)^2+2x(1-x)\frac{\tr(P^2)+\tr(Q^2)}{2}\le 1,
\end{eqnarray}
where the first inequality follows from the Cauchy-Schwarz inequality. Hence all the inequalities are saturated only if $P$ is linearly dependent with $Q^\G$ and $\tr(Q^\G)=\tr(Q)=1$. This implies that $W$ is the partial transpose of a pure entangled state. 

(iii) Recall from (i) that $W_{a,b,0,0}$ is a legitimate normalized DEW for $b\in(0,1]$, with its smallest eigenvalue given by $\l_{mn}(W_{a,b,,0,0})=-\frac{b}{2}$. Hence for every $x\in [-\frac{1}{2},0)$, $W_{a,b,,0,0}$ exists such that $\l_{mn}(W_{a,b,,0,0})=x$.
Combining with Lemma \ref{use} (ii), the infimum of $\l_{mn}$ is $-\frac{1}{2}$ and attainable, the supremum of $\l_{mn}$ is $0$ and not attainable.

(iv) The if part can be verified straightforwardly. 
We prove the only if part. Suppose $W$ is a normalized DEW such that $\l_{mn}(W)=-\frac{1}{2}$. Using Lemma \ref{opdc} and Lemma \ref{use} (ii), $W$ can be written as $W=Q^\Gamma$, with unit-trace positive operator $Q$ whose range is a CES.  Considering the spectral decomposition $Q=\sum_{i} p_i\proj{\phi_i}$, we have
\begin{eqnarray}
	\label{efg}
	-\frac{1}{2}=\l_{mn}(Q^\Gamma)\ge \sum_{i} p_i\l_{mn}(\proj{\phi_i}^\Gamma)\ge -\frac{1}{2}(\sum_{i} p_i)=-\frac{1}{2},
\end{eqnarray}
where the first inequality follows from Lemma \ref{ineqe} (i) and the second follows from Lemma \ref{ept}. Hence, the two inequalities in (\ref{efg}) are both saturated. We obtain that $\l_{mn}(\proj{\phi_i}^\Gamma)=-\frac{1}{2}$, which implies the Schmidt coefficients of each $\ket{\phi_i}$ are $(\frac{\sqrt{2}}{2},\frac{\sqrt{2}}{2},0\cdots,0)$. Write $\ket{\phi_i}=\frac{\sqrt{2}}{2}(\ket{b_1^{(i)}}\otimes \ket{c_1^{(i)}}+\ket{b_2^{(i)}}\otimes \ket{c_2^{(i)}})$ as the Schmidt decomposition. 
Since the first inequality is saturated, using Lemma \ref{ineqe} (ii), we obtain that the one-dimensional eigenspace corresponding to the eigenvalue $-\frac{1}{2}$ of each $\proj{\phi_i}^\Gamma$ is identical. By Lemma \ref{ept}, this can be formulated as 
\begin{eqnarray}
	\label{span}
	\ket{{b_1^{(i)}}^*}\otimes \ket{{c_2^{(i)}}}-\ket{{b_2^{(i)}}^*}\otimes \ket{{c_1^{(i)}}}=t(\ket{{b_1^{(1)}}^*}\otimes \ket{{c_2^{(1)}}}-\ket{{b_2^{(1)}}^*}\otimes \ket{{c_1^{(1)}}})
\end{eqnarray}
for each $i\ge 2$ and some $t$ with $|t|=1$. From (\ref{span}), we obtain that
$\ket{b_1^{(i)}},\ket{b_2^{(i)}}\in \span\{\ket{b_1^{(1)}},\ket{b_2^{(1)}}\}$ and $\ket{c_1^{(i)}},\ket{c_2^{(i)}}\in \span\{\ket{c_1^{(1)}},\ket{c_2^{(1)}}\}$.
Hence, $\span\{\ket{\phi_1},\ket{\phi_2}\}$ is a subspace of $\span\{\ket{b_1^{(1)}},\ket{b_2^{(1)}}\}\otimes \span\{\ket{c_1^{(1)}},\ket{c_2^{(1)}}\}$, which is isomorphism to $\bbC^2\otimes \bbC^2$.
If $\text{dim}(\span\{\ket{\phi_1},\ket{\phi_2}\})=2$, then by Lemma \ref{pv} (i), it must contain a product vector. This contradicts the fact that the range of $Q$ is a CES. We conclude that $\ket{\phi_{2}}$ is linearly dependent with $\ket{\phi_{1}}$, and similarly, all $\ket{\phi_i}$ are linearly dependent with $\ket{\phi_{1}}$. So we have $W=Q^\G=\proj{\phi_1}^\Gamma$.
This proves the only if part. 

(v) Recall from (i) that $(W_{a,b,,0,0})$ is a legitimate DEW whenever $b\in (0,1]$. We have $\l_1(W_{a,b,,0,0})={1-b\over mn-1}+{b\over2}$, and it  monotonically increases with $b$. Hence for each $x\in (\frac{1}{mn-1},\frac{1}{2}]$, there exists $W_{a,b,,0,0}$ such that $\l_1(W_{a,b,,0,0})=x$. To cover the remaining interval $(\frac{1}{2},1)$, we consider the alternative construction
$W_{0,b,c,0}$ in (\ref{abcd}). For $b\in (0,1]$, it is a legitimate DEW, with $\l_1(W_{0,b,c,0})=1-\frac{b}{2}\in [\frac{1}{2},1)$. Combining these two cases, we conclude that the claim holds. By Lemma \ref{use} (iii), 
the infimum and supremum of $\l_1$ are  $\frac{1}{mn-1}$ and 1 respectively, and they are both unattainable.

(vi) The inequality $\cN(W)>0$ is by definition of EWs.
For arbitrary pure state $\ket{\phi}\in \bbC^m\otimes \bbC^n$ with Schmidt coefficients $a_1\ge\cdots\ge a_m\ge 0$, from Lemma \ref{ept}, we have
\begin{eqnarray}
	\label{inpr}
	\cN(\proj{\phi}^\Gamma)=\sum_{1\le i<j\le 1} |-a_ia_j|\le \sum_{1\le i<j\le 1} \frac{a_i^2+a_j^2}{2}={m-1\over2},
\end{eqnarray}
where the inequality is saturated if and only if all the Schmidt coefficients of $\ket{\phi}$ are identical, meaning $\ket{\phi}$ is the maximally entangled state. Write $W=P+Q^\G$ and let the spectral decomposition of $Q$ be $Q=\sum_{i} p_i\proj{\phi_i}$. We have
\begin{eqnarray}
	\label{nega}
	\cN (P+Q^\Gamma)\le \cN (Q^\Gamma)\le \sum_i p_i\cN (\proj{\phi_i}^\Gamma)\le \tr(Q^\Gamma)\cdot{m-1\over2}\le {m-1\over2}.
\end{eqnarray}
Here, the first inequality follows from Lemma \ref{ineqe} (ii), the second inequality follows from the convexity of negativity, and the third inequality follows from (\ref{inpr}). This proves the first claim. 
Consider $W_{0,0,c,d}$ in (\ref{abcd}). It is a legitimate DEW for $d\in (0,1]$ with $\cN(W_{0,0,c,d})=\frac{(m-1)d}{2}$. Thus, for each $x\in (0,\frac{m-1}{2}]$, $W$ exists such that $\cN(W)=x$. We conclude that the infimum of $\cN(W)$ is 0 and not attainable, the supremum of $\cN(W)$ is $\frac{m-1}{2}$ and attainable.

(vii) The if part can be verified by direct computation. We prove the only if part. 

We first prove the case that $m=n$. Suppose $W$ is an $m\times m$ normalized DEW such that $\cN(W)=\frac{m-1}{2}$. Then all the inequalities in (\ref{nega}) are saturated. This implies that $P=0$ and each $\ket{\phi_i}$ is a maximally entangled state. Thus, $\l(\proj{\psi_i}^\Gamma)$ are identical for each $i$, which are $\frac{1}{m}$ with multiplicity $\frac{m(m+1)}{2}$ and $-\frac{1}{m}$ with multiplicity $\frac{m(m-1)}{2}$. 
Assume that $W$ has $t$ negative eigenvalues. If $t<\frac{m(m-1)}{2}$, then from Lemma \ref{ineqe} (i), we have $\sum_{j=0}^{t-1}\l_{mn-j}(W)\ge \sum_{i=1}^k p_i (\sum_{j=0}^{t-1}\l_{mn-j}(\proj{\phi_i}^\Gamma))=-\frac{t}{m}>\frac{(1-m)}{2}$. This contradicts with $\cN(W)=\frac{m-1}{2}$.
Similarly, we obtain that $t$ cannot be greater than $\frac{m(m-1)}{2}$. So $W$ has exactly $\frac{m(m-1)}{2}$ negative eigenvalues. By applying Lemma \ref{ineqe} (i) again,  $\l_{\frac{m(m+1)}{2}+1}(W)=\cdots=\l_{mn}(W)=-\frac{1}{m}$ and 
consequently, $\l_1(W)=\cdots=\l_{\frac{m(m+1)}{2}}(W)=\frac{1}{m}$.
Hence $\l(W)=\sum_{i=1}^k p_i\l(\proj{\phi_i}^\Gamma)=\l(\proj{\phi_i}^\Gamma)$. It follows that all $\proj{\phi_i}^\Gamma$ share common eigenvectors and thus commute with each other. We conclude that they are linearly dependent. Hence $W=Q^\G=\proj{\phi_1}^\G$ and the claim holds.

We next prove the general case $m<n$.
Let $W$ be a normalized DEW such that $\cN(W)=\frac{m-1}{2}$. Similarly, we obtain that $W=\sum_{i}  q_i(\proj{\psi_i})^\G$, where each $\ket{\psi_i}$ is a maximally entangled state. Under local unitary equivalence, we can assume  $\ket{\psi_1}=\ket{\Phi_1}$.
Suppose $W$ has $s$ negative eigenvalues with corresponding unit eigenvectors $\ket{\a_1},\cdots,\ket{\a_s}$. We have 
\begin{eqnarray}
	\cN(W)=-\sum_i q_i(\sum_{j=1}^s \bra{\a_j}(\proj{\psi_i}^\G) \ket{\a_j})\le \frac{m-1}{2}.
\end{eqnarray}
The saturation of the inequality implies that 
\begin{eqnarray}
	\label{A8}
\sum_{j=1}^s \bra{\a_j} (\proj{\psi_i}^\G)\ket{\a_j}=\frac{m-1}{2}
\end{eqnarray}
holds for each $i$. Focusing on $i=1$, (\ref{A8}) implies $\frac{m(m-1)}{2}$ eigenvectors among $\ket{a_i}$ correspond to the $\frac{m(m-1)}{2}$ negative eigenvalues of $\ket{\Phi_1}$, which implies $s\ge\frac{m(m-1)}{2}$. Without loss of generality, let these be $\ket{\a_1},\cdots,\ket{\a_{\frac{m(m-1)}{2}}}$. If $s=\frac{m(m-1)}{2}$, then (\ref{A8}) implies that all $\ket{\psi_i}$ belong to $\bbC^m\otimes \span\{\ket{1},\cdots,\ket{m}\}$ and $W$ can be viewed as an $m\times m$ DEW. By the previous claim, all $\ket{\psi_i}$ are linearly dependent with $\ket{\psi_1}$ and the claim holds. 

On the other hand, if $s>\frac{m(m-1)}{2}$, then by applying (\ref{A8}) again, we have 
\begin{eqnarray}
	\bra{\a_j}(\proj{\Phi_1}^\G)\ket{\a_j}=0
\end{eqnarray}
holds for any $\frac{m(m-1)}{2}+1\le j\le s$. It follows from the expression of $\ket{\Phi_1}$ that $\ket{\a_{\frac{m(m-1)}{2}+1}},\cdots,\ket{\a_s}\in \bbC^m \otimes \span\{\ket{m+1},\cdots,\ket{n}\}$. 
Consequently, Lemma \ref{ept} ensures that $\ket{\psi_i}$ must lie in $\bbC^m \otimes \span\{\ket{m+1},\cdots,\ket{n}\}$ for $i\ge 2$. Under local equivalence, we can set $\ket{\psi_2}=\ket{\Phi_2}$, while keeping $\ket{\psi_1}$ unchanged. Repeating this argument inductively, we can write $\ket{\psi_i}=\ket{\Phi_i}$ for any $i$ under local unitary uquivalence. Therefore, $W$ can be written as the convex combination of $\proj{\Phi_1}^\G,\cdots,\proj{\Phi_k}^\G$. The upper bound $k\le \lfloor \frac{n}{m} \rfloor$ arises because $\span\{\ket{\Phi_1},\cdots,\ket{\Phi_k}\}\subseteq \bbC^m\otimes \bbC^n$. 
This completes the proof.

(viii) Lemma \ref{use} (iv.b) indicates that $\sum_{i=3}^{2n} \l_i$ is lower bounded by $-\frac{1}{2+2\sqrt{2}}$. The attainment of the bound can be directly verified by using Lemma \ref{ept}. Similarly, it follows from Lemma \ref{use} (iv.c) that $\sum_{i=k}^{2n} \l_{i}$ is lower bounded by $-\frac{1}{2}$ for any $4\le k\le 2n$. The attainment of this bound can also be directly verified.

(ix) Let $U$ be the order-$mn$ unitary matrix such that $UWU^\dg=\diag(\l_{mn},\cdots,\l_1)$. Recalling from Lemma \ref{twoap} that $\r_1\in \app_{m,n}$ and thus $U^\dg \r U$ is PPT, since $W$ is decomposable, we have 
\begin{eqnarray}
	\label{genhao2}
\tr(UWU^\dg\cdot\r_1)=\tr(W\cdot U^\dg\r_1U)=\sqrt{2}(\l_{mn}+\l_{mn-1})+1\ge 0,
\end{eqnarray}
which proves the claim. Similarly, the second claim can be proved by using Lemma \ref{twoap} that $\r_2\in \app_{m,n}$. The attainments of the infima 
can be directly verified.

\section{Proof of Theorem \ref{NDEW}}
\label{ndd}
This section contains the proof of Theorem \ref{NDEW}. We first propose a necessary lemma.

\begin{lemma}
	\label{o2}
	(i)	Let $A_1 = \begin{pmatrix} \frac{\sqrt{2}}{2} & 0 \\ 0 & \frac{\sqrt{2}}{2} \end{pmatrix}$ and $B_1 = \begin{pmatrix} 0 & x \\ y & 0 \end{pmatrix}$, where $|x|,|y| \leq \frac{\sqrt{2}}{2}$ are not all zero. Then there exist $p_1, p_2\in \mathbb{C} \setminus \{0\}$ satisfying $|p_1|^2 + |p_2|^2 = 1$ such that the maximal singular value of the matrix $C_1 = p_1 A_1 + p_2 B_1$ is strictly greater than $\frac{\sqrt{2}}{2}$.
	
	(ii) Let $A_2=\begin{pmatrix} \frac{\sqrt{2}}{2} & 0 \\ 0 & 0 \end{pmatrix}$ and $B_2=\begin{pmatrix} 0 & \frac{\sqrt{2}}{2} \\ z & 0 \end{pmatrix}$, where $|z| \leq \frac{\sqrt{2}}{2}$ is not zero. Then there exist $p_1, p_2\in \mathbb{C} \setminus \{0\}$ satisfying $|p_1|^2 + |p_2|^2 = 1$ such that the maximal singular value of the matrix $C_2= p_1 A_2 + p_2 B_2$ is strictly greater than $\frac{\sqrt{2}}{2}$.
\end{lemma}

\begin{proof}
	(i)	By computation,	we have
	\begin{eqnarray}
		\lambda_{1}(C_1^\dg C_1) = \frac{u + v + \sqrt{(u - v)^2 + 4 |z|^2}}{2},
	\end{eqnarray}
	where
	\begin{eqnarray}
		u: &&= \frac{|p_1|^2}{2} + |p_2|^2 |y|^2, \\
		v: &&= \frac{|p_1|^2}{2} + |p_2|^2 |x|^2,\\
		z: &&= \frac{\sqrt{2}}{2}(p_1p_2^*y^*+p_1^*p_2x).
	\end{eqnarray}
	Let $t:= |p_2|^2$. By choosing appropriate phases for $p_1$ and $p_2$, we can maximize $|z|$ such that 
	$|z| = \sqrt{t(1 - t)} \cdot \frac{\sqrt{2}}{2} (|x| + |y|)$.
	To have $\sigma_{\max}(C_1) > \frac{\sqrt{2}}{2}$, we require 
	\begin{eqnarray}
		u + v + \sqrt{(u - v)^2 + 4 |z|^2}>1,
	\end{eqnarray}
	by substituting the expressions, we have
	\begin{eqnarray}
		\label{newi}
		t(|x|^2 + |y|^2 - 1) + \sqrt{t^2(|x|^2 - |y|^2)^2 + 2 t(1 - t)(|x|+|y|)^2} > 0.
	\end{eqnarray}
	For sufficiently small $t > 0$, one can verify that (\ref{newi}) holds as 
	the square root term dominates. We conclude that the claim holds.
	
	(ii) The proof is similar to that of (i) through directly analyzing the maximal eigenvalue of $C_2^\dg C_2$.
\end{proof}

\vspace{0.3cm}

{\bf Proof of Theorem \ref{NDEW}}

(i) Firstly, Lemma \ref{use} (i) implies that $\tr(W^2)\le 1$.  Using Lemma \ref{use2}, there exists a normalized NDEW $W$ such that $\l(W)$ can be arbitrarily close to $(1,0,\cdots,0)$ (corresponds to a pure state). Therefore $\tr(W^2)$ can be arbitrarily close to 1. This implies that the supremum of $\tr(W^2)$ is 1.

(ii)  Let $t$ be the infimum of $\l_1$ for normalized $m\times n$ NDEWs. Suppose $t$ is attained by an NDEW $W$. By Lemma \ref{use} (iii), we have $t>\frac{1}{mn-1}$. For any PPT entangled state $\rho$ detected by $W$, the perturbed witness $W':=\frac{1}{1+\epsilon mn}(W+\epsilon I_{mn})$, with sufficiently small $\epsilon > 0$ remains a normalized NDEW detecting $\rho$. Moreover, a direct computation implies that $\tr(W'^2)<t$, yielding a contradiction. Thus, the infimum cannot be attained.

(iii)
Following Lemma \ref{use2}, there exists an NDEW $W$ whose eigenvalue vector can be arbitrarily close to $(1,0,\cdots,0)$, which implies that $\l_{mn}(W)$ can be arbitrarily close to 0. Combining with Lemma \ref{use} (ii) that $\l_{mn}(W)<0$, we conclude that the supremum of $\l_{mn}(W)$ is 0 and unattainable.

(iv) Following Lemma \ref{use} (ii) that $\l_{mn}\ge -\frac{1}{2}$, it suffices to prove the existence of normalized NDEWs with smallest eigenvalues approaching $-\frac{1}{2}$. For the qubit-qudit system, consider the PPT edge state 
\begin{eqnarray}
	\label{rhob}
	\r_b=\frac{1}{7b+1}\bma b&0&0&0&0&b&0&0\\0&b&0&0&0&0&b&0\\0&0&b&0&0&0&0&b\\0&0&0&b&0&0&0&0\\0&0&0&0&\frac{1+b}{2}&0&0&\frac{\sqrt{1-b^2}}{2}\\b&0&0&0&0&b&0&0\\0&b&0&0&0&0&b&0\\0&0&b&0&\frac{\sqrt{1-b^2}}{2}&0&0&\frac{1+b}{2}\ema,
\end{eqnarray}
where $b\in (0,1)$ \cite{horodecki1997separability}. 
Let $P$ and $Q$ be the projectors onto the kernel of $\r_b$ and $\r_b^\G$, respectively. Given any $z>0$, define $\epsilon_z:=\min\limits_{\ket{e,f},\norm{e}=\norm{f}=1}
\bra{e,f}(zP+Q^\G)\ket{e,f}>0$, where the positivity of $\epsilon_z$ follows from the edge state properties \cite{lewenstein2001characterization}. Construct 
\begin{eqnarray}
	W_z:=\frac{1}{z\tr(P)+\tr(Q^\Gamma)-8\delta}(zP+Q^\G-\delta I_8),
\end{eqnarray}
where $0<\delta<\epsilon_z$. We have $W_z$ is a normalized NDEW as it can detect the entanglement of $\r_b$. Next, a direct computation yields $\lim\limits_{b\rightarrow1} \l_{8}(\frac{1}{\tr(Q^\G)}Q^\G)=-\frac{1}{2}$. As $b\to 1$ and $z,\delta \to 0$, Lemma \ref{ineqe} (iii) shows $\l_{8}(W_z)$ approaches $-\frac{1}{2}$.

For the two-qutrit system, we analogously consider the PPT edge state 
\begin{eqnarray}
	\label{rhoa}
	\r_a=\frac{1}{8a+1}\bma
	a & 0 & 0 & 0 & a & 0 & 0 & 0 & a \\
	0 & a & 0 & 0 & 0 & 0 & 0 & 0 & 0 \\
	0 & 0 & a & 0 & 0 & 0 & 0 & 0 & 0 \\
	0 & 0 & 0 & a & 0 & 0 & 0 & 0 & 0 \\
	a & 0 & 0 & 0 & a & 0 & 0 & 0 & a \\
	0 & 0 & 0 & 0 & 0 & a & 0 & 0 & 0 \\
	0 & 0 & 0 & 0 & 0 & 0 & \frac{a+1}{2} & 0 & \frac{\sqrt{1-a^2}}{2} \\
	0 & 0 & 0 & 0 & 0 & 0 & 0 & a & 0 \\
	a & 0 & 0 & 0 & a & 0 & \frac{\sqrt{1-a^2}}{2} & 0 & \frac{a+1}{2} \\
	\ema,
\end{eqnarray}
where $a\in (0,1)$ \cite{horodecki1997separability}.
Let $P'$ and $Q'$ be the projectors onto the kernel of $\r_a$ and $\r_a^\G$. A direct computation also yields $\lim\limits_{a\rightarrow1} \l_{9}(\frac{1}{\tr(Q'^\G)}Q'^\G)=-\frac{1}{2}$. 
Following a similar construction to the qubit-qudit case, we can build a normalized NDEW $W'_z$ with $\l_9(W'_z)$ approaching $-\frac{1}{2}$ in the appropriate limit. This result naturally extends to arbitrary higher-dimensional systems through zero-padding of $W_z$ or $W'_z$.
This proves that the infimum of $\l_{mn}$ is $-\frac{1}{2}$.

Finally, we prove that the infimum is not attainable. Assume $W$ is an $m\times n$ normalized NDEW such that $\l_{mn}(W)=-\frac{1}{2}$, with the coresponding eigenvector being $\ket{r}$.  Up to local unitary equivalence, we can assume $\ket{r}=\sum^k_{j=1}r_j\ket{jj}$, where $r_j>0$ and $k\le \min\{m,n\}$. Firstly, suppose $k=2$. We project $W$ onto the subspace spanned by $\{\ket{1},\ket{2}\}\otimes \{\ket{1},\ket{2}\}$ and obtain a two-qubit EW $W'$. One can verify that the smallest eigenvalue of $W'$ is $-\frac{1}{2}$, with the same eigenvector $\ket{r}$. Since $\tr(W')\le \tr(W)=1$,
by Lemma \ref{use} (ii), we have $\tr(W')=1$. Consequently, all the diagonal entries except those of $W'$ are zero. From Lemma \ref{proj}, this implies that all rows and columns beyond $W'$ are zero. We conclude that $W$ is indeed decomposable, which leads to a contradiction.

On the other hand, suppose $k\ge 3$. We have $r_ir_j<\frac{1}{2}$ for $i\neq j$. Let $\r':=\sum\limits_{i\neq j} r_ir_j\proj{ij}$. Since $W\ket{r}=-\frac{1}{2}\ket{r}$, we have
\begin{eqnarray}
	\label{dd1}
	\tr(W\cdot(\proj{r}+\frac{1}{2}I_{mn}))
	=\tr(W\cdot(\proj{r}+\r'))+\tr(W\cdot(\frac{1}{2}I_{mn}-\r'))=0.
\end{eqnarray}
According to the results from Appendix B of \cite{Vidal1999Robustness}, we know that the unnormalized state $\proj{r}+\r'$ is separable. Hence (\ref{dd1}) implies that $\tr(W(\frac{1}{2}I_{mn}-\r'))=0$.
The condition $r_ir_j<\frac{1}{2}$ implies that $\frac{1}{2}I_{mn}-\r'$ is a strictly positive diagonal matrix. Hence all the diagonal entries of $W$ are zero, implying $W$ is a zero matrix. This is a contradiction.
We conclude that the infimum cannot be attained.

(v) Using Lemma \ref{use2}, there exists NDEW $W$ whose eigenvalue vector  can be arbitrarily close to $(1,0,\cdots,0)$, which implies that $\l_{1}(W)$ can be arbitrarily close to 1. Combining with Lemma \ref{use} (iii) that $\l_1<1$, we conclude that the supremum of $\l_1(W)$ is 1 and is not attainable.

(vi) Let $t$ be the infimum of $\l_1$ for normalized $m\times n$ NDEWs. Suppose $t$ is attained by some $W$. By Lemma \ref{use} (iii), we have $t>\frac{1}{mn-1}$. For any PPT entangled state $\rho$ detected by $W$, the perturbed witness $W':=\frac{1}{1+\epsilon mn}(W+\epsilon I_{mn})$, with sufficiently small $\epsilon > 0$ remains a normalized NDEW detecting $\rho$. Moreover, it satisfies that
$\l_1(W')=t+\frac{(1-t mn)\epsilon}{1+mn\epsilon}<t$, yielding a contradiction. Thus, the infimum cannot be attained.

(vii) Let $W$ be a normalized $2\times n$ NDEW $W$ with $l$ negative eigenvalues. From Lemma \ref{pv} (ii), we have $1\le l\le n-1$. It then follows from Lemma \ref{use} (iv) that $\sum_{i=2n-l+1}^{2n}\l_i(W)\ge -\frac{1}{2}$. This implies that the negativity of $W$ is not greater than $\frac{1}{2}$.
Meanwhile, result (iv) demonstrates that $\l_{2n}$ can be arbitrarily close to $-\frac{1}{2}$. This proves the claim for $m=2$.

For $m=n\ge 3$, we have known from Lemma \ref{use} (v) that $\cN(W)\le \frac{m-1}{2}$. It then suffices to prove that there exists a normalized NDEW whose negativity can be arbitrarily close to $\frac{m-1}{2}$.
Consider the following two-qutrit PPT entangled state  \cite{clarisse2006construction1}
\begin{eqnarray}
	\label{ga3}
	\g=\frac{1}{13}\bma 1&0&0&0&0&0&0&0&-1\\ 0&2&0&-1&0&0&0&0&0\\0&0&1&0&0&0&1&0&0\\0&-1&0&1&0&0&0&0&1\\0&0&0&0&1&0&1&0&0\\0&0&0&0&0&1&0&-1&0\\0&0&1&0&1&0&2&0&0\\0&0&0&0&0&-1&0&1&0\\-1&0&0&1&0&0&0&0&3\ema.
\end{eqnarray}
Let $U=\bma -1&0&0\\0&-1&0\\0&0&1\ema$, $V=\bma 0&0&1\\0&1&0\\1&0&0\ema$.
The state $\g':=((U\otimes V)\g(U\otimes V)^\dg)^\G$ is also a PPT entangled state. Moreover, a direct computation yields $\tr(\proj{\Psi_m}^\G\cdot\g')=0$, where $\ket{\Psi_m}:=\frac{1}{\sqrt{m}}\sum_{i=1}^m \ket{ii}\in \bbC^m\otimes \bbC^m$ is the maximally entangled state. 
Let $W_{\g'}$ be a $m\times m$ normalized NDEW detecting $\g'$. Define $W_{t,\g'}:=\frac{1}{1+t}(t\proj{\Psi_m}^\G+W_{\g'})$  where $t>0$. We have $W_{t,\g'}$ is still an NDEW that detects $\g'$. Taking $t\to \infty$ and applying Lemma \ref{ineqe}(iii), we have $\l(W_{t,\g'})$ converges to  $\l(\proj{\Psi_m}^\G)$, whose negativity is $\frac{m-1}{2}$. This can be generalized to arbitrary $m\times m$ systems by just appending zero rows and columns on $W_{t,\g'}$.
We conclude that the claim holds.

(viii) Let $W$ be a $3\times 3$ normalized EW such that $\cN(W)=\frac{m-1}{2}=1$. We shall show that $W$ must be decomposable.
Firstly, by Lemma \ref{use} (v), six eigenvalues of $W$ are $\frac{1}{3}$ and three are $-\frac{1}{3}$. We write 
\begin{eqnarray}
	W=\frac{1}{3}(I_9-2P):=\frac{1}{3}(I_9-2(\proj{a}+\proj{b}+\proj{c})),
\end{eqnarray}
where $P$ is the rank-three projection onto the negative eigenspace. Up to a unitary transformation of the basis, we may assume $\ket{a}$ has Schmidt rank two. Moreover, by the block-positivity of $W$, we have $\bra{x,y}P\ket{x,y}\le \frac{1}{2}$ for any unit vectors $\ket{x},\ket{y}\in \bbC^3$. This implies that the maximal Schmidt coefficient of any unit vector in the range of $P$ is at most $\frac{\sqrt{2}}{2}$.
Hence, $\ket{a}$ has Schmidt coefficeints $(\frac{\sqrt{2}}{2},\frac{\sqrt{2}}{2},0)$. Under local unitary equivalence, we may take
\begin{eqnarray}
	\label{a}
	\ket{a}=\frac{\sqrt{2}}{2}(\ket{11}+\ket{22}).
\end{eqnarray}
Let $A:=\bma\frac{\sqrt{2}}{2}&&0&&0\\0&&\frac{\sqrt{2}}{2}&&0\\0&&0&&0\ema$,
$B:=[b_{i,j}]$ and $C:=[c_{i,j}]$ be the coefficient matrices of $\ket{a},\ket{b}$ and $\ket{c}$ with respect to the standard basis. The above condition implies that the largest singular value of
\begin{eqnarray}
	\label{X}
	X:=p_1A +p_2 B+p_3 C
\end{eqnarray}
is at most $\frac{\sqrt{2}}{2}$ for any $|p_1|^2+|p_2|^2+|p_3|^2=1$.
In particular, this implies that $|b_{i,j}|,|c_{i,j}|\le \frac{\sqrt{2}}{2}$ for all $i,j$. On the other hand, since $\bra{11}W\ket{11}=-\frac{2}{3}|b_{1,1}|^2-\frac{2}{3}|c_{1,1}|^2\ge 0$, we have 
$b_{1,1}=c_{1,1}=0,
$
and similarly,
$
b_{2,2}=c_{2,2}=0.
$
Now, assume at least one of $b_{1,2}$ and $b_{2,1}$ is nonzero.
By using Lemma \ref{o2} (i), there exist nonzero $p_1,p_2$ (with $p_3=0$) such that the maximal singular value of the top-left order-two matrix of $X$ is greater than $\frac{\sqrt{2}}{2}$. This implies that the maximal singular value of $X$ also exceeds $\frac{\sqrt{2}}{2}$, which leads to a contradiction. So we have 
$
b_{1,2}=b_{2,1}=0,
$
and similarly, 
$
c_{1,2}=c_{2,1}=0.
$
Hence, both $B$ and $C$ have rank two. A direct computation gives the two nonzero singular values of $B$ are
\begin{eqnarray}
	\frac{\sqrt{2}}{2}\pm \sqrt{1-4(|b_{1,3}|^2+|b_{2,3}|^2)(|b_{3,1}|^2+|b_{3,2}|^2)},
\end{eqnarray}
which implies 
\begin{eqnarray}
	(|b_{1,3}|^2+|b_{2,3}|^2)(|b_{3,1}|^2+|b_{3,2}|^2)=\frac{1}{4}.
\end{eqnarray}
Note that $(|b_{1,3}|^2+|b_{2,3}|^2)(|b_{3,1}|^2+|b_{3,2}|^2)\le \frac{1}{4}(|b_{1,3}|^2+|b_{2,3}|^2+|b_{3,1}|^2+|b_{3,2}|^2)\le \frac{1}{4}$, equality holds only if
\begin{eqnarray}
	\label{bbcc}
	b_{3,3}=0,|b_{1,3}|^2+|b_{2,3}|^2=|b_{3,1}|^2+|b_{3,2}|^2=\frac{1}{2}.
\end{eqnarray}
Similarly, we have
\begin{eqnarray}
	\label{bbccc}
	c_{3,3}=0,|c_{1,3}|^2+|c_{2,3}|^2=|c_{3,1}|^2+|c_{3,2}|^2=\frac{1}{2}.
\end{eqnarray}
By applying a unitary transformation on $\ket{b}$ and $\ket{c}$, we may set 
$
b_{1,3}=0.
$
It follows from (\ref{bbcc}) that $|b_{2,3}|=\frac{\sqrt{2}}{2}$. Up to a 
global phase, we may assume
$
b_{2,3}=\frac{\sqrt{2}}{2}.
$
Assume $b_{3,2}\neq 0$, then by Lemma \ref{o2} (ii), there exist nonzero $p_1,p_2$ (with $p_3=0$) such that the lower-right order-two matrix of $X$ in (\ref{X}) is greater than $\frac{\sqrt{2}}{2}$, again a contradiction. So we have
$b_{3,2}=0$ and thus $|b_{3,1}|=\frac{\sqrt{2}}{2}$. By applying another local unitary transformation without changing $\ket{a}$ and $b_{2,3}$, we can further eliminate the phase of $b_{3,1}$ such that 
\begin{eqnarray}
	\label{b}
	\ket{b}=\frac{\sqrt{2}}{2}(\ket{23}+\ket{31}).
\end{eqnarray}
It then follows from the block-positivity of $W$ that $c_{2,3}=c_{3,1}=0$. So $|c_{1,3}|=|c_{3,2}|=\frac{\sqrt{2}}{2}$ holds by (\ref{bbccc}). Up to a global phase of $\ket{c}$, assume that $c_{1,3}=\frac{\sqrt{2}}{2}$ and $c_{3,2}=e^{i\theta}\frac{\sqrt{2}}{2}$. We claim that $e^{i\theta}=-1$. To see this, consider $p_1=p_2=p_3=\frac{\sqrt{3}}{3}$ for $X$ in (\ref{X}). The characteristic polynomial of $X^\dg X$ is
\begin{eqnarray}
	\label{lll}
	f(\l):=\det(\l I_3-X^\dg X)=\frac{1}{108}(108 \l^3-108 \l^2+27\l-1-\cos\theta),
\end{eqnarray}
where $f(\frac{1}{2})=-(1+\cos \theta)\le 0$ and $f(\l)>0$ for sufficiently large $\l$.
If $-(1+\cos \theta)<0$, then there is a root of $f(\l)=0$ that is greater than $\frac{1}{2}$, which implies that at least one singular value of $X$ is greater than $\frac{\sqrt{2}}{2}$. This is a contradiction.  We conclude that $\cos\theta=-1$ and therefore 
\begin{eqnarray}
	\label{c}
	\ket{c}=\frac{\sqrt{2}}{2}(\ket{13}-\ket{32}).
\end{eqnarray}

Finally, 
one can verify that from (\ref{a}), (\ref{b}) and (\ref{c}) that
$W^\G=\frac{1}{3}(I_9-2(\proj{a}^\G+\proj{b}^\G+\proj{c}^\G))\ge 0$, which would lead to a contradiction if $W$ is nondecomposable. We conclude that the supremum of $\cN(W)$ is not attainable by normalized NDEWs for $m=n=3$.

(ix)  Using Lemma \ref{use2}, there exists NDEW $W$ whose eigenvalue vector  can be arbitrarily close to $(1,0,\cdots,0)$, which implies that $\cN(W)$ can be arbitrarily close to 0. On the other hand, $\cN(W)>0$ by the definitions of EWs. 
We conclude that the infimum of $\cN(W)$ is 0 and not attainable.
$\hfill\square$

\section{Proof of Theorem \ref{NPTNDEW}}
\label{adem}
By definition, any PPT entangled state can be detected by an NDEW. It suffices to prove that any NPT state can also be detected by an NDEW. 
Let
$\r$ be an $m\times n$ ($m\le n$) NPT entangled state, and $\ket{\psi}$ be the eigenvector corresponding to the smallest eigenvalue of $\r^\G$. Suppose $\ket{\psi}$ has Schmidt rank $d$. We have $2\le d\le m$.
There exist invertible matrices $A$ and $B$ such that $(A\otimes B)\ket{\Psi_d}=\ket{\psi}$, where $\ket{\Psi_d}:=\frac{1}{\sqrt{d}}\sum_{i=1}^d \ket{ii}\in \bbC^m\otimes \bbC^n$.
Let $\r':=(A^T\otimes B)^\dg\r (A^T\otimes B)$.
So $\r'$ is also an $m\times n$ NPT entangled state, and it satisfies that
\begin{eqnarray}
	\label{nd1}
	\tr(\proj{\Psi_d}^\G\cdot\r')
	=\tr(\proj{\psi}^\G\cdot\r)<0.
\end{eqnarray}
Under local equivalence, it then suffices to prove that $\r'$ can be detected by an NDEW.

We begin with $m=2$. In this case, $d$ is exactly two. 
Consider the state $\sigma:=(\diag(-1,1)\otimes I_4)\r_b^\G(\diag(-1,1)\otimes I_4)$, where $\r_b$ in (\ref{rhob}) is a $2\times 4$ PPT entangled state. Hence $\sigma$ is also a $2\times 4$ PPT entangled state, and by appending zero rows and columns, a $2\times n$ PPT entangled state. So there exists a $2\times n$ NDEW 
$W_\sigma$ detecting it. On the other hand, a direct computation yields $\tr(\sigma\cdot \proj{\Psi_2}^\G)=0$. For any $t>0$, define $W_{t,\sigma}:=t\proj{\Psi_2}^\G+W_\sigma$. We obtain that $W_{t,\sigma}$ is an NDEW since it remains block-positive and detects $\sigma$. Further, 
\begin{eqnarray}
	\label{nd2}
	\tr(W_{t,\sigma}\cdot \r')=t\tr(\proj{\Psi_2}^\G\cdot\r')+\tr(W_\sigma\cdot \r').
\end{eqnarray}
Let $t$ be large enough, according to (\ref{nd1}), the rhs of (\ref{nd2}) becomes negative. This implies that $\r'$ can be detected by a $W_{t,\sigma}$.

We next consider $m\ge 3$. Recall from (\ref{ga3}) that the state $\gamma$ is a  two-qutrit PPT entangled state.
Let $U_1=\bma 0&1&0\\0&0&1\\1&0&0\ema$, $U_2=\bma -1&0&0\\0&-1&0\\0&0&1\ema$ and $V=\bma 0&0&1\\0&1&0\\1&0&0\ema$. Let $\g_1:=((U_1\otimes V)\g(U_1\otimes V)^\dg)^\G$ and $\g_2:=((U_2\otimes V)\g(U_2\otimes V)^\dg)^\G$. One can verify that $\g_1$ and $\g_2$ are both $m\times n$ PPT entangled states
by appending zero rows and columns. Further, direct computations yield 
\begin{eqnarray}
	\label{hmy1}
&&\tr(\proj{\Psi_2}^\G\cdot\g_1)=0,\\
\label{hmy2}
&&\tr(\proj{\Psi_d}^\G\cdot\g_2)=0
\end{eqnarray}
for any $d\ge 3$. Let $W_{\g_1}$  and $W_{\g_2}$ be the NDEWs detecting $\g_1$ and $\g_2$, respectively. For any $t>0$, define $W_{t,\g_1}:=t\proj{\Psi_2}^\G+W_{\g_1}$ and $W_{t,\g_2}:=t\proj{\Psi_d}^\G+W_{\g_2}$. 
According to (\ref{hmy1}) and (\ref{hmy2}), $W_{t,\g_1}$ and $W_{t,\g_2}$ are both NDEWs that detect $\g_1$ and $\g_2$ respectively. If the Schmidt rank $d$ is two, then (\ref{nd1}) implies that $\r'$ can be detected by a $W_{t,\g_1}$ by letting $t$ be large enough. Similarly, if the Schmidt rank $d$ is no smaller than three, then $\r'$ can be detected by an NDEW $W_{t,\g_2}$.
We conclude that the claim holds.
$\hfill\square$

\bibliographystyle{unsrt}
\bibliography{Spectralcharacterizationsofew}

\end{document}